\documentclass[12pt]{article}

\RequirePackage[OT1]{fontenc}
\usepackage{amsthm,amsmath,amssymb,amsfonts,bm,enumerate,xcolor,graphicx,natbib,color,url}
\RequirePackage[colorlinks,citecolor=blue,urlcolor=blue]{hyperref}

\usepackage{threeparttable}

\def\frac#1#2{{\textstyle{#1\over#2}}}

\DeclareSymbolFont{AMSb}{U}{msb}{m}{n}
\DeclareMathSymbol{\Natural}{\mathbin}{AMSb}{"4E}
\DeclareMathSymbol{\Integer}{\mathbin}{AMSb}{"5A}
\DeclareMathSymbol{\Real}{\mathbin}{AMSb}{"52}
\DeclareMathSymbol{\Rational}{\mathbin}{AMSb}{"51}
\DeclareMathSymbol{\Imaginary}{\mathbin}{AMSb}{"49}
\DeclareMathSymbol{\Complex}{\mathbin}{AMSb}{"43} 
\DeclareMathSymbol{\Disk}{\mathbin}{AMSb}{"44} 
\def\bi{\begin{itemize}}
\def\ei{\end{itemize}}
\def\bd{\begin{description}}
\def\ed{\end{description}}
\def\ben{\begin{enumerate}}
\def\een{\end{enumerate}}
\def\bv{\begin{verbatim}}
\def\ev{\end{verbatim}}

\def\calD{{\mathcal D}}

\def\calN{{\mathcal N}}

\def\calP{{\mathcal{P}}}

\def\calS{{\mathcal{S}}}
\def\calI{{\mathcal{I}}}

\def\hat#1{{\widehat{#1}}}

\def\pr{{\rm Pr}}
\def\Pr{\pr}
\def\E{{\rm E}}

\def\var{{\rm var}}

\def\Dto{{\ {\buildrel D\over \longrightarrow}\ }}

\def\2to{{\ {\buildrel 2\over \longrightarrow}\ }}

\def\Deq{{\ {\buildrel D\over =}\ }}

\def\I1ton{{$I_1,\ldots,I_n$}}
\def\X1ton{{$X_1,\ldots,X_n$}}
\def\Y1ton{{$Y_1,\ldots,Y_n$}}
\def\Z1ton{{$Z_1,\ldots,Z_n$}}
\def\R1ton{{$R_1,\ldots,R_n$}}
\def\e1ton{{$e_1,\ldots,e_n$}}
\def\t1ton{{$t_1,\ldots,t_n$}}
\def\x1ton{{$x_1,\ldots,x_n$}}
\def\y1ton{{$y_1,\ldots,y_n$}}
\def\z1ton{{$z_1,\ldots,z_n$}}
\def\calP{{\mathcal{P}}}
\def\calS{{\mathcal{S}}}
\def\calI{{\mathcal{I}}}

\usepackage[english]{babel}
\usepackage{amsfonts}
\usepackage{amssymb}
\usepackage{graphicx}
\usepackage{enumerate}

\usepackage{placeins}
\newcommand*{\thead}[1]{\multicolumn{1}{c}{\bfseries #1}}

\usepackage{ulem}

\def\calS{{\mathcal{S}}}
\def\pr{{\rm Pr}}
\def\Pr{\pr}
\def\E{{\rm E}}
\def\calP{{\mathcal{P}}}
\def\calD{{\mathcal D}}

\def\Deq{{\ {\buildrel D\over =}\ }}

\def\calN{{\mathcal N}}
\def\calI{{\mathcal{I}}}
\def\Dto{{\ {\buildrel D\over \longrightarrow}\ }}
\def\var{{\rm var}}

\newcommand{\blind}{1}

\newtheorem{proposition}{Proposition}[section]

\begin{document}
%%%%%%%%%%%%%%%%%%%%%%%%%%%%%%%%%%%%%%%%%%%%%%%%%%%%%%%%%%%%%%%%%%%%%%%%%%%%%%
\thispagestyle{empty}
\baselineskip=28pt
\vskip 5mm
\begin{center} {\Large{\bf Max-infinitely divisible models and inference for spatial extremes}}
\end{center}

\baselineskip=12pt

%{\noindent \color{red} THOMAS: More generally (thinking beyond the scope of this paper), I am wondering if there is a different physical interpretation of random scales and shifts that could help motivate one or the other beyond the rather pragmatic arguments concerning their theoretical and practical tractability. one might say that either the mean behavior (random shift) or the scale (random scale) is more erratic in the individual extreme events compared to a Gaussian limiting process for central tendencies. it would be interesting to know if there is more theory on this (random effects in mean or scale) in the mixed model literature.  
%\noindent RAPHAEL: Which title do you prefer? the new one above, or the old one: ``Gaussian Scale Mixtures for Spatial Extremes: From Asymptotic Independence to Asymptotic Dependence''?}

\vskip 5mm

\if1\blind
{
\begin{center}
\large
Rapha\"el Huser$^1$, Thomas Opitz$^2$ and Emeric Thibaud$^3$%\\ %\emph{{\small\color{red}We'll see later for the order of authors, depending on our contributions...}}
\end{center}

\footnotetext[1]{
\baselineskip=10pt Computer, Electrical and Mathematical Sciences and Engineering (CEMSE) Division, King Abdullah University of Science and Technology (KAUST), Thuwal 23955-6900, Saudi Arabia. E-mail: raphael.huser@kaust.edu.sa}
\footnotetext[2]{
\baselineskip=10pt  INRAE, UR546 Biostatistics and Spatial Processes, 228, Route de l'A\'erodrome, CS 40509, 84914 Avignon, France. E-mail: thomas.opitz@inrae.fr}
\footnotetext[3]{
\baselineskip=10pt Ecole Polytechnique F\'ed\'erale de Lausanne, EPFL-FSB-MATHAA-STAT, Station 8, 1015 Lausanne, Switzerland. E-mail: emeric.thibaud@epfl.ch}
} \fi

\baselineskip=17pt
\vskip 4mm
\centerline{\today}
\vskip 6mm

%%%%%%%%%%%%%%%%%%%%%%%%%%%%%%%%%%%%%%%%%%%%%%%%%%%%%%%%%%%%%%%%%%%%%%%%
\begin{center}
{\large{\bf Abstract}}
\end{center}
%new abstract (<=200 words):
For many environmental processes, recent studies have shown that the dependence strength is decreasing when quantile levels increase. This implies that the popular max-stable models are inadequate to capture the rate of joint tail decay, and to estimate joint extremal probabilities beyond observed levels. We here develop a more flexible modeling framework based on the class of max-infinitely divisible processes, which extend max-stable processes while retaining dependence properties that are natural for maxima. We propose two parametric constructions for max-infinitely divisible models, which relax the max-stability property but remain close to some popular max-stable models obtained as special cases. The first model considers maxima over a finite, random number of independent observations, while the second model generalizes the spectral representation of max-stable processes. Inference is performed using a pairwise likelihood. We illustrate the benefits of our new modeling framework on Dutch wind gust maxima calculated over different time units. Results strongly suggest that our proposed models outperform other natural models, such as the Student-t copula process and its max-stable limit, even for large block sizes.
\baselineskip=16pt

\par\vfill\noindent

%%\clearpage\pagebreak\newpage 
%\pagenumbering{arabic}
%%\baselineskip=24pt
%\baselineskip=22pt
%
%\if1\blind
%{
%  \title{\bf Penultimate modeling  of spatial  extremes: \\ statistical inference for max-infinitely divisible processes}
%  \author{Rapha\"{e}l Huser\footnote{We'll later see for the order of authors, depending on our contributions...}\\
%    CEMSE Division, King Abdullah University of Science and Technology,\\ 23955-6900, Thuwal, Saudi Arabia,\\
%    Thomas Opitz\\
%    BioSP, INRA, 84914, Avignon, France\\
%    and\\
%    Emeric Thibaud\\
%    Ecole Polytechnique F\'ed\'erale de Lausanne, Switzerland
%    }
%  \maketitle
%} \fi
%
%\if0\blind
%{
%  \bigskip
%  \bigskip
%  \bigskip
%  \begin{center}
%    {\LARGE\bf Penultimate modeling  of spatial  extremes: \\ statistical inference for max-infinitely divisible processes}
%\end{center}
%  \medskip
%} \fi
%
%\bigskip
%\begin{abstract}
%
%\end{abstract}
%
%\noindent%
%{\it Keywords:} asymptotic dependence and independence; block maximum; extreme event; max-infinitely divisible process; sub-asymptotic modeling.\\
%\vfill

%%%%%%%%%%%%%%%%%%%%%%%%%%%%%%%%%%%%%%%%%%%%%%%%%%%%%%%%%%%%%%%%%%%%%%%%
%%%%%%%%%%%%%%%%%%%%%%%%%%%%%%%%%%%%%%%%%%%%%%%%%%%%%%%%%%%%%%%%%%%%%%%%

\newpage
%\spacingset{1.45} % DON'T change the spacing! 
\baselineskip=20pt

\section{Introduction}
Spatial extreme-value analysis for georeferenced observations in environmental or climatological studies aims at producing statistical inferences and predictions for events in the tail of the distribution, potentially much more extreme than the observed events, while appropriately taking into account their spatial dependence structure \citep[see, e.g.,][]{Reich.Shaby:2012a,Huser.Davison.2014}. In this context, max-stable processes have emerged as useful models for spatial extremes \citep{Davison.etal.2012,Davison.etal:2019}. Thanks to their asymptotic characterization for spatial block maxima, max-stable processes are usually fitted to maxima observed over temporal blocks at a given set of locations. In practice, the choice of a suitable block size depends on practical aspects, such as the presence of natural cycles (e.g.,  seasonal components), the frequency of observations and the temporal dependence strength. Statistically, this implies a bias-variance trade-off: larger blocks typically yield a better representation of the data's tail properties, but also mean that fewer maxima are available for {inference}, thus inflating the estimation uncertainty. Therefore, we always need to choose a finite, and often relatively small, block size, which casts doubts on the validity of the max-stability assumption in practice. A fast growing body of empirical studies on environmental extremes in the literature has indeed revealed that the max-stability assumption arising asymptotically is often violated at finite levels \citep{Bopp.al.2020}, and that the spatial dependence strength is often weakening as events become more extreme \citep[see, e.g.,][]{Huser.etal:2017,Tawn.al.2018,Huser.Wadsworth:2019,Bacro.al.2019,CastroCamilo.Huser:2019,CastroCamilo.al.2020}. In particular, when data are asymptotically independent, maxima become ultimately independent at the highest levels, requiring specialized models capturing the rate of decay towards independence.

To illustrate this phenomenon with real data, consider the following intuitive dependence summary. Let $\bm Z=(Z_1,\ldots,Z_D)^{\rm T}$ be a $D$-dimension random vector with joint distribution $F$, assumed to be continuous for simplicity, with marginal distributions $F_1,\ldots,F_D$. We define the level-dependent coefficient $\theta_D(u)$, for probability level $u\in(0,1)$, as
\begin{equation}\label{eq:depcoef}
\theta_D(u) = { { \log\left[F\left\{F_1^{-1}(u),\ldots, F_D^{-1}(u)\right\}\right]  }\over {  \log (u)  } },\quad u\in (0,1).
\end{equation}
Notice that the underlying \emph{copula} structure defined as $C(u_1,\ldots,u_D):=F\{F_1^{-1}(u_1),\ldots, F_D^{-1}(u_D)\}$ (with margins transformed to the uniform ${\rm Unif}(0,1)$ scale) satisfies the relation $C(u,\ldots,u) = u^{\theta_D(u)}$, $u\in[0,1]$, so that $\theta_D(u)$ can be interpreted as the effective number of independent variables \emph{at the level $u$}. When the variables $Z_1,\ldots,Z_D$ are perfectly dependent, we get $\theta_D(u)=1$; when they are positive quadrant-dependent, we get $\theta_D(u)\in[1,D]$; when they are independent, we get $\theta_D(u)=D$; and when they are negative quadrant-dependent, we get $\theta_D(u)\geq D$. The coefficient \eqref{eq:depcoef} is a copula-based measure of association, and its limit as $u\to0$ was introduced in \citet{Hua.Joe:2011} as the lower tail order. Moreover, when $F$ is a max-stable distribution, $\theta_D(u)$ is constant in $u$ and lies within the interval $[1,D]$, and it is known as the \emph{extremal coefficient} \citep{Davison.Huser:2015,Strokorb.al.2015}. In the following, we will therefore refer to \eqref{eq:depcoef} as the \emph{level-dependent extremal coefficient}. To illustrate the limitations of max-stable models in a sub-asymptotic setting, Figure~\ref{fig:extcoefWinds} shows an empirical estimate of $\theta_D(u)$ in \eqref{eq:depcoef}, plotted with respect to $u\in(0,1)$, for the Dutch wind speed dataset analyzed in Section \ref{sec:application}, here for $D=30$ monitoring stations, calculated over daily, weekly and monthly blocks. Figure~\ref{fig:extcoefWinds}  reveals that $\theta_D(u)$ is increasing rather than being constant with respect to $u$, especially near the upper tail, and that $\theta_D(u)$ is larger for larger block sizes. This contradicts max-stability for these block sizes, and suggests that the dependence strength weakens as events become more extreme. More flexible models are needed to adequately capture this data feature.

\begin{figure}%[bt]
	\centering
	\includegraphics[width=0.7\linewidth]{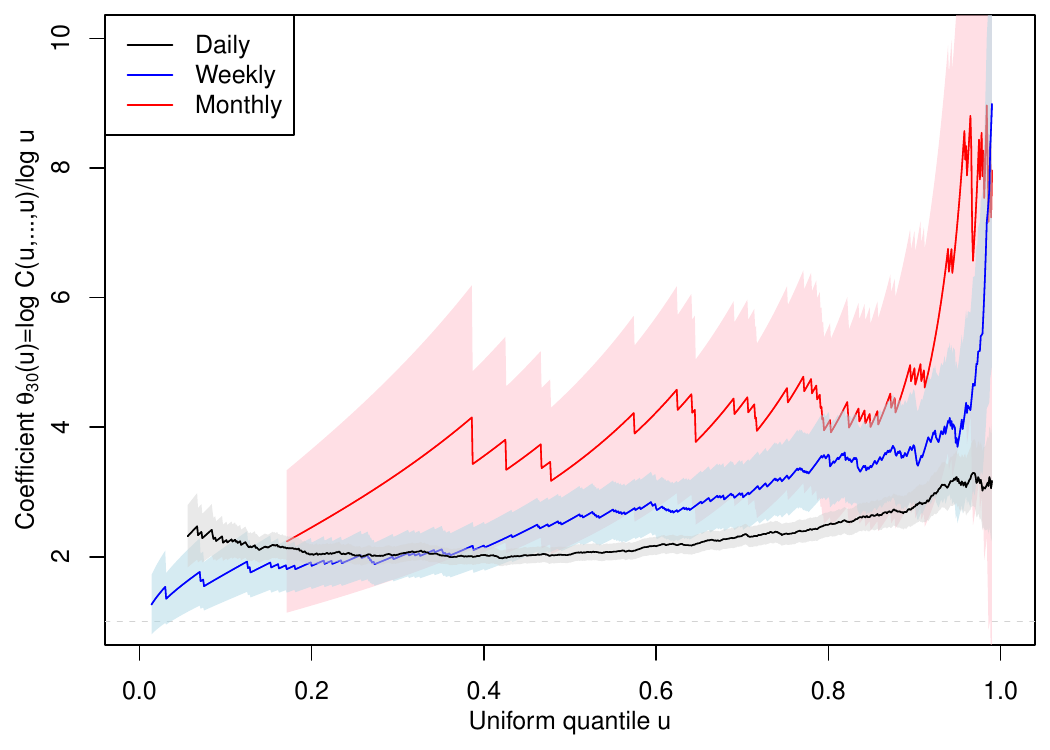}
	\caption{{Empirical level-dependent extremal coefficient $\hat\theta_D^{\rm emp}(u)$} (solid curves) for the Dutch wind speed data analyzed in the Application section~\ref{sec:application}, plotted as a function of the {uniform quantile $u$}, based on daily (black), weekly (blue) and monthly (red) maxima. Shaded areas are $95\%$-confidence intervals based on the delta method, ignoring weak temporal dependence.}
	\label{fig:extcoefWinds}
\end{figure}

In this paper, we suggest using max-infinitely divisible (max-id) processes as models for capturing dependence in finite block maxima. Max-id models play an important role in the limit theory of triangular arrays of random vectors \citep{Balkema.al.1993} and naturally extend the class of max-stable models, relaxing their restrictive stability properties. They retain attractive theoretical properties reflecting the particular positive dependence structure of maxima resulting from the operation of taking pointwise maxima of independent and identically distributed (iid) processes, and they ensure validity of distributions of maxima of events after a change of temporal support, e.g., when characterizing the joint tail of monthly variables based on a distribution fitted to yearly maxima. The latter property holds if we assume that the original events possess a stationary distribution in time  and arrive independently in time according to a Poisson process. Such assumptions may not always be realistic in practice, but these properties nevertheless  suggest using max-id distributions as conceptually appealing models preserving mechanisms that are well known for max-stable processes.

As the class of max-id process is substantially larger than that of max-stable processes, it is important to constrain max-id models in a sensible way. Here, we propose general construction principles for building new parametric max-id models that remain in the ``neighborhood'' of some popular max-stable processes, with tail dependence properties that can be precisely characterized. We focus on designing max-id models for spatial maxima that smoothly bridge asymptotic dependence and independence regimes, while keeping the widely-used extremal-$t$ max-stable process \citep{op2013} as a max-stable submodel on the boundary of the parameter space. The increased tail flexibility of our new models  makes them attractive for modeling maxima taken over relatively small blocks such as days or weeks, which also increases the effective sample size when fitting the model, leading to improved statistical efficiency. Our focus is on spatial modeling, but of course distributions for multivariate random vectors represent a special case of this framework. 

Theoretical properties of max-id processes have been explored in depth \citep[see, e.g.,][]{Gine.al.1990,Dombry.EyiMinko.2013}. However, applications beyond max-stability are rare. \citet{Padoan.2013} proposed a max-id model based on Gaussian process ratios, whose dependence strength varies with the intensity of the extreme event. Moreover, a recent implementation of a specific spatial max-id model using simulation-based Bayesian inference  \citep{Bopp.al.2020} bears witness of the versatility of such models. Compared to the models proposed by \citet{Padoan.2013} and \citet{Bopp.al.2020}, our proposed max-id models can control the joint tail decay rate more flexibly, and unlike \citet{Padoan.2013}, they can be arbitrarily close to popular max-stable models. Here, we propose new construction principles for building new max-id processes. In particular, we offer a wide class of new max-id models that extend the spectral representation of max-stable processes. As monotone increasing marginal transformations do not affect the max-id structure, we propose using the generalized extreme-value (GEV) distribution for univariate margins.

The paper is organized as follows. In Section \ref{sec:maxid}, we recall some theory on max-id distributions. In Section \ref{sec:model}, we discuss two general construction principles for max-id processes, and we investigate the theoretical tail properties of our new parametric models, showing that they can capture asymptotic dependence and/or independence. Likelihood-based inference is then presented in Section \ref{sec:inference} and a simulation study is discussed in Section \ref{sec:simulation}. We finally describe an application to wind gust data in Section \ref{sec:application}, before concluding with some remarks and perspectives on future research in Section \ref{sec:discussion}.

%%%%%%%%%%%%%%%%%%%%%%%%%%%%%%%%%%%%
%%%%%% SECTION 2: MAX-ID DISTRIBUTIONS %%%%%%%%%%
%%%%%%%%%%%%%%%%%%%%%%%%%%%%%%%%%%%%

\section{Max-infinitely divisible distributions}\label{sec:maxid}
\subsection{Definition and Poisson process construction}
A multivariate distribution function $G$ on $\mathbb{R}^D$, $D \geq 1$,  is max-id if and only if $G^t$ (using the notation $G^t(\bm z)=\{G(\bm z)\}^t$) is a valid distribution function on $\mathbb{R}^D$ for any $t>0$. In particular, a max-id distribution $G$ describes the componentwise maximum of $m$ independent random variables with distribution $F=G^{1/m}$, for any $m=1,2,\ldots$.  In practice, this property permits to switch from the joint distribution $G$ of the componentwise maximum over a given time unit to alternative time units and in particular to the distribution $F$ of the original events. For example, by fitting a max-id model to annual maxima of a variable of interest, conclusions may be drawn for monthly, weekly or  daily maxima, modulo non-stationary and temporal dependence aspects. Unlike the univariate case, multivariate distributions are not always max-id. 

To propose useful max-id models, we will exploit a constructive characterisation of max-id distributions based on Poisson processes \citep[Chapter 5]{Resnick.1987}. For simplicity, we provide the following definitions and characterisations in terms of multivariate distributions and refer to \citet{Gine.al.1990} and \citet{Kabluchko.Schlather.2010,Kabluchko.al.2016} for the generalization to stochastic processes. In the following, we identify the Euclidean space $\mathbb{R}^D$ with its topological closure $[-\infty,\infty]^D$. 

We consider a Poisson point process (PPP) defined on the domain $E=[l_1,\infty]\times\ldots\times[l_D,\infty] \subseteq \mathbb{R}^D$ with mean measure $\Lambda\geq 0$, where the lower endpoint $\bm l=(l_1,\ldots,l_D)^{\rm T}\in [-\infty,\infty)^D$ may contain components that are equal to $-\infty$. For theoretical reasons, we impose that $\Lambda$ must be a Radon measure on $E\setminus\{\bm l\}$ such that Borel sets $A$ with infinite mass $\Lambda(A)$ may only arise when the closure of $A$ contains the lower endpoint $\bm l$. The Poisson points are defined as
\begin{equation}\label{eq:PPP}
\{\bm X_i;\, i=1,\dots,N\} \sim \textsc{PPP}(\Lambda), \quad N\in \mathbb{N}_0 \cup \{\infty\},
\end{equation}
and we extend the measure $\Lambda$ to $\mathbb{R}^D$ by setting $\Lambda(A)=\Lambda(A\cap E)$ for all Borel sets $A$ of $\mathbb{R}^D$. We then define a random vector $\bm Z=(Z_1,\ldots,Z_D)^{\rm T}\in\mathbb{R}^D$ with support contained in $E$ as the componentwise maximum over the Poisson points $\bm X_i$ and the lower endpoint $\bm l$, i.e.,
\begin{equation}\label{eq:constructmaxid}
\bm Z:=\max\big(\max_{i=1,2,\ldots,N} \bm X_i, \bm l\big).
\end{equation}
The value $\bm Z=\bm l$ arises when $N=0$, i.e., when the Poisson process contains no points in $E\setminus \{\bm l\}$. From \citet[Proposition 5.8]{Resnick.1987}, it follows that $\bm Z$ is max-id, and its joint distribution function is 
\begin{equation}\label{eq:cdf}
G(\bm z)=\exp\left\{-\Lambda\left([-\bm\infty,\bm z]^{\rm c}\right)\right\}, \ \bm z\in E, \qquad G(\bm z)=0, \  \bm z\in E^{\rm c},
\end{equation}
where $\bm z=(z_1,\ldots,z_D)^{\rm T}$, $[-\bm\infty,\bm z]=[-\infty,z_1]\times\cdots\times[-\infty,z_D]$, and $A^{\rm c}$ is the complement of the set $A$ in $\mathbb{R}^D$. Result \eqref{eq:cdf} can be easily derived by noticing that the event $\{\bm Z\leq \bm z\}$, $\bm z\in E$, is equivalent to having no points $\bm X_i$ in $\left[-\bm\infty,\bm z\right]^{\rm c}$. The measure $\Lambda$ is called the exponent measure or mean measure of $G$, and $V(\bm z)=\Lambda(\left[-\bm\infty,\bm z\right]^{\rm c})$ is called the exponent function. To ensure that $\bm Z$ has no components with value $+\infty$, we further impose the theoretical restriction $\Lambda(\{\bm z\in E:\max_{j=1,\ldots,D} z_j=\infty\})=0$. Univariate distributions of the max-id vector are given by 
\begin{equation*}
G_j(z)=\exp\left\{-\Lambda_j\left([-\infty,z]^{\rm c}\right)\right\}, \  z\in [l_j,\infty],\  j=1,\ldots,D,
\end{equation*}
and $G_j(z)=0$ otherwise, where $\Lambda_j\left([-\infty,z]^{\rm c}\right)=\Lambda\left(\{\bm z\in E\mid z_j > z\}\right)$ denotes the $j$-th margin of  $\Lambda$ ($j=1,\ldots,D$).

The Poisson process representation based on \eqref{eq:PPP} and \eqref{eq:constructmaxid} is helpful for intuitive interpretation (Section \ref{sec:spectralrep}), modeling (Section \ref{sec:finitemeasure}--\ref{sec:spectralrep}) and simulation (Section~\ref{sec:sim}). 

From the definition, it is easy to verify that any univariate distribution function is max-id. Moreover, if a random vector $\bm Z=(Z_1,\ldots,Z_D)^{\rm T}$ is max-id, the marginally transformed vector $\{h_1(Z_1),\ldots,h_D(Z_D)\}^{\rm T}$ with nondecreasing functions $h_j$ $(j=1,\ldots,D)$ remains max-id. Therefore, the max-id property concerns primarily the dependence structure, not the margins. When assuming continuous margins for simplicity, the underlying copula, obtained as 
\begin{equation}\label{eq:copula}
C(\bm u)=\exp\left\{-\Lambda^\star\left([\bm 0,\bm u]^{\rm c}\right)\right\}, \ \bm u=(u_1,\ldots,u_D)^{\rm T}\in[0,1]^D,
\end{equation}
where $\bm 0=(0,\ldots,0)^{\rm T}$ is the $D$-dimensional vector of zeros and $\Lambda^\star\left([\bm 0,\bm u]^{\rm c}\right) = \Lambda\left([-\bm\infty,\bm z]^{\rm c}\right)$ with $\bm z=(z_1,\ldots,z_D)^{\rm T}$ such that $z_j=G_j^{-1}(u_j)$, is  thus max-id itself. 
Here, the $j$-th margin satisfies $\Lambda_j^\star([0,u]^{\rm c})=\Lambda^\star(\{\bm u\in[0,1]^D: u_j>u\})=-\log u$, such that the max-id vector $\bm U=(U_1,\ldots,U_D)^{\rm T}\sim C$ with $U_j=G_j(Z_j)$ satisfies $\pr(U_j\leq u)=u$, $u\in[0,1]$ ($j=1,\ldots,D$). When the margins are not continuous (e.g., if $G$ has a point mass at the lower bound), slight modifications are required in \eqref{eq:copula}. In the following, we will write $G$ and $\Lambda$ for a max-id distribution and its underlying exponent measure, respectively, on any arbitrary marginal scale, and we will write $C$ and $\Lambda^\star$ for the underlying max-id copula and its corresponding exponent measure, respectively. Unless specified otherwise, we will assume that $G$ has continuous margins.

Any random vector $\bm Z$ with independent components $Z_j$ is max-id with exponent measure concentrated on the half-axes $\{l_1\}\times\ldots \times\{l_{j-1}\}\times [l_j,\infty]\times \{l_{j+1}\}\times\ldots\times \{l_D\}$ $(j=1,\ldots,D)$. The corresponding copula is $C(\bm u)=\prod_{j=1}^D u_j$. Similarly, any fully dependent random vector, with copula $C(\bm u)=\min(u_1,\ldots,u_D)$, is also max-id. Section~\ref{sec:depprop} investigates further dependence properties.

\subsection{Dependence properties}\label{sec:depprop}
Max-id random vectors $\bm Z$ are associated \citep[Proposition 5.29]{Resnick.1987}, such that a certain form of positive dependence prevails. Thus, negatively correlated random vectors cannot be max-id; see the Supplementary Material for a bivariate Gaussian counter-example that also illustrates the effect of the choice of $t$ in $G^t$. 

Extremal dependence is closely related to the tail behavior of the exponent measure $\Lambda$ since 
\begin{equation}\label{eq:approxtail}
1-G(\bm z)=1-\exp\left\{-\Lambda\left([-\bm\infty,\bm z]^{\rm c}\right)\right\}\sim\Lambda\left([-\bm\infty,\bm z]^{\rm c}\right), \quad \min_{j=1,\ldots,D} z_j\rightarrow\infty. 
\end{equation}
If a max-id distribution $G$ with exponent measure $\Lambda$ is used to model the componentwise maximum over $m$ independent random vectors with distribution $F$ such that $F^m=G$, then 
\begin{equation}\label{eq:tailFn}
F(\bm z) = G^{1/m}(\bm z) = \exp\left\{-\Lambda\left([-\bm\infty,\bm z]^{\rm c}\right)/m\right\}, 
\end{equation}
which gives the first-order tail approximation $1-F(\bm z)\approx \Lambda([-\bm\infty,\bm z]^{\rm c})/m$ when $\bm z$ has large components and/or $m$ is large. This shows that the extremal dependence structures of $F$, $G$ and $\Lambda$ are of the same form. 
%\thom{[WE HAVE TO BE CAREFUL WITH RESPECT TO THE SINGULAR MASS AT THE LOWER ENDPOINT ARISING FOR FINITE EXPONENT MEASURES. IN THIS CASE, THE LOWER ENDPOINT ON THE COPULA SCALE IS NOT 0, BUT RATHER EXP(-MASS AT 0). TO AVOID TECHNICAL DETAILS, WE COULD SIMPLY RESTRICT THE PRESENTATION TO THE INFINITE MEASURE CASE, AND SAY THAT OTHERWISE SOME CORRECTIONS MUST BE DONE IN THE FORMULAS FOR THE LOWER ENDPOINT.]}

The dependence strength may be summarized by the level-dependent extremal coefficient $\theta_D(u)$ as defined in \eqref{eq:depcoef}. In the case of max-id distributions, this coefficient can be expressed in terms of the underlying exponent measure as $\theta_D(u)= -\Lambda^\star\left([\bm 0,u\bm 1]^{\rm c}\right)/\log u$, where $\bm 1=(1,\ldots,1)^{\rm T}$ is a $D$-dimensional vector of ones. This coefficient can be interpreted, at the probability level $u\in(0,1)$, as the equivalent number of independent variables amongst the $D$ variables. A similar dependence summary was considered by \citet{Padoan.2013}. Furthermore, for bivariate vectors $(U_1,U_2)^{\rm T}\sim C$ with uniform margins and max-id copula $C$, we have that
\begin{align}
\label{chiz}
\chi(u)=\pr(U_1>u\mid U_2>u) &\; =2-{1-C(u,u)\over 1-u} \sim 2-\theta_2(u),\quad u\to1.
\end{align}
When $\chi=\lim_{u\to1}\chi(u)=0$, which occurs when $\theta_2(u)\to2$, %or more strongly when $\theta_D(u)\to D$, 
as $u\to1$, the pair of variables $(U_1,U_2)^{\rm T}$ is called asymptotically independent, whilst they are asymptotically dependent if $\chi>0$. Asymptotic independence reflects scenarios of weakening and ultimately vanishing dependence strength as a function of the quantile level $u$. 

\subsection{Max-stable distributions}
Max-stable processes have been widely used for modeling spatial extremes \citep{Davison.etal.2012}. If a joint distribution $F$ is such that for some sequences of vectors $\bm a_m\in\mathbb R_+^D$ and $\bm b_m\in\mathbb R^D$,
\begin{equation}\label{eq:convmax}
F^m(\bm a_mz+\bm b_m)\to G(\bm z),\qquad m\to\infty,
\end{equation} 
where the distribution $G$ has non-degenerate margins, then the limit $G$ is max-stable. Essentially, the max-stability property implies that for any $t>0$, $G^t$ is a valid distribution that is of the same type as $G$ itself (i.e., $G$ and $G^t$ only differ through marginal location and scale parameters); see \citet{Davison.etal.2012} for a more precise definition. Max-stable processes form a subclass within the class of max-id processes. Indeed, if we allow $F=F_m$ to depend on $m$ in the convergence \eqref{eq:convmax}, then the limit distribution $G$ is max-id but not necessarily max-stable \citep{Balkema.al.1993}. If $G$ is max-stable with unit Fr\'echet marginal distributions, i.e., $G_j(z)=\exp(-1/z)$, $z>0$ ($j=1,\ldots,D$), then we have $E=[0,\infty]^D$, $\Lambda_j([0,z]^{\rm c})=1/z$, and $t\Lambda(tA)=\Lambda(A)$ for all $t>0$ and Borel set $A$, such that the extremal coefficient $\theta_D(u)\equiv\theta_D$ defined in \eqref{eq:depcoef} is constant with level $u\in(0,1)$. Therefore, max-stable models can only capture asymptotic dependence or exact independence, but fail at representing weakening dependence; recall \eqref{chiz}. This is often too strong an assumption for environmental data (see Figure~\ref{fig:extcoefWinds}). The broader class of max-id models allows us to gain in flexibility by relaxing this stability requirement. 

By analogy with \eqref{eq:constructmaxid}, max-stable processes with unit Fr\'echet margins are often defined constructively through their spectral representation
\begin{equation}\label{eq:ms}
Z(\bm s)=\max_{i=1,2,\ldots} R_i W_i(\bm s), \qquad \bm s\in\calS\subset\mathbb{R}^d,
\end{equation}
where $\{R_i\}$ are the points of a Poisson process on $\mathbb{R}_+$ with intensity $r^{-2}\,\mathrm{d}r$, and $W_i(\bm s)$ are independent copies of a random process $W(\bm s)$ with $\E[\max\{W(\bm s),0\}]=1$, independent of $\{R_i\}$; see \citet{deHaan:1984} and \citet{Schlather:2002}. Observe that for $Z_j:=Z(\bm s_j)$ ($j=1,\ldots,D$), this is a special case of \eqref{eq:constructmaxid} with a suitable choice of exponent measure $\Lambda$. Moreover, the total mass of the measure $\Lambda$ is here infinite, which removes the effect of the lower boundary in \eqref{eq:constructmaxid}.

%%%%%%%%%%%%%%%%%%%%%%%%%%%%%%%%
%%%%%% SECTION 3 : MODELING %%%%%%%%%%%%
%%%%%%%%%%%%%%%%%%%%%%%%%%%%%%%%

\section{Modeling}\label{sec:model}

\subsection{Construction principles}\label{sec:constrprinc}
We broadly distinguish three approaches to building useful max-id models: either by (i) directly specifying the exponent measure $\Lambda$ in \eqref{eq:PPP} and \eqref{eq:cdf}, or (ii) defining the points $\bm X_i$ constructively in the representation \eqref{eq:constructmaxid}, or (iii) exploiting the fact that max-id distributions arise as limits of $F_m^m$ as $m\to\infty$ where the distributions $F_m$ are not necessarily identical. This last approach was used by \citet{Padoan.2013}, who obtained a max-id model as the limit of multivariate Gaussian ratios with correlation increasing with $m$. Here, we propose two new general construction principles: in Section \ref{sec:finitemeasure}, we follow (i) by defining a \emph{finite} exponent measure $\Lambda$, while in Section \ref{sec:infinitemeasure}, we follow (ii) and directly define the points $\bm X_i$  in \eqref{eq:constructmaxid}, generalizing the spectral representation of max-stable processes in \eqref{eq:ms} with \emph{infinite} exponent measure $\Lambda$.

\subsection{Models with finite exponent measure $\Lambda$}\label{sec:finitemeasure}
Using a finite exponent measure $\Lambda=cH$ parametrized by an arbitrary probability distribution $H$ on $E$ and a constant $c>0$, the  max-id vector $\bm Z$ in \eqref{eq:constructmaxid} has joint distribution 
\begin{equation}\label{eq:finitemeasuremodel}
G_{c,H}(\bm z)=\exp[-c\{1-H(\bm z)\}],\ \bm z\in E, \qquad G(\bm z)=0, \  \bm z\in E^{\rm c}.
\end{equation}
From the construction \eqref{eq:constructmaxid}, $\bm Z$ can be interpreted as the componentwise maximum over a finite number $N$ of independent events, where $N$ follows the Poisson distribution with mean $c$. To simulate the max-id vector $\bm Z$, we first sample $N\sim{\rm Poisson}(c)$, then conditionally generate $\bm X_1,\ldots,\bm X_N$ independently from $H$, and finally set $\bm Z=\max(\bm X_1,\ldots,\bm X_N,\bm l)$. As $\Lambda$ is finite and the event $\{N=0\}$ has probability $\exp(-c)>0$, this yields positive mass at the lower boundary $\bm l$. In practice, this singularity is rather a nuisance than a relevant model feature, and we may restrict $c$ to the range $[c_0,\infty)$ with a relatively large value of $c_0>0$, to ensure that  $\exp(-c)\approx 0$. Once a parametric model for $H$ has been chosen, the additional parameter $c$ refines the tail behavior of $G$ as compared to that of $H$ and adds flexibility. Consider the distribution $F=G_{c,H}^{1/m}$ of the original observations, for some fixed $m>0$. Using \eqref{eq:tailFn},
\begin{equation}\label{eq:facjoint}
1-F(\bm z)=1-G_{c,H}^{1/m}(\bm z) = 1-\exp[-(c/m)\{1-H(\bm z)\}] \sim (c/m)\{1-H(\bm z)\},
\end{equation}
as $m\to\infty$ and/or $\min_{j=1,\ldots,D} z_j\rightarrow\infty$, so that the constant $c$ controls the tail weight of $F$ with respect to that of $H$. Using the approximation \eqref{eq:facjoint} in \eqref{chiz} shows that the asymptotic dependence class of $F$ and $H$ is the same: the limit value of $\chi$ in \eqref{chiz} is the same for $F$ and $H$. More generally, in the case where a max-stable limit exists, both $H$ and the max-id distribution $G_{c,H}$ are in the maximum domain of attraction of the same max-stable limit process (modulo parameters in the univariate limit distributions); this result follows from multivariate regular variation theory \citep[\emph{e.g.}, Proposition 5.17 of ][]{Resnick.1987} where the constant $c$ cancels out in limit expressions. %The max-id variants obtained for various values of $c>0$ can be interpreted as an interpolation 
The finite measure model $G_{c,H}$ in \eqref{eq:finitemeasuremodel} therefore interpolates between the tail behavior of $H$ (for $c=1$) and that of the max-stable limit (for $c\to\infty$), since it corresponds to taking the componentwise maximum over $N$ independent processes (conditionally on $N>0$), and $N/c$ converges to $1$ almost surely as $c$ tends to infinity. In other words, the parameter $c$ controls ``how far'' the model is from its max-stable counterpart obtained as a special limiting case as $c\to\infty$, which is useful when modeling sub-asymptotic block maxima. Moreover, this gives us a way of constructing new asymptotically (in)dependent max-id models from essentially arbitrary distributions $H$ with the same {tail dependence characteristics. In contrast to using directly $H^t$, for some real $t>0$, as a model, our approach ensures a well-defined distribution $F=G_{c,H}^{1/m}$ for any $m=1,2,\ldots$, even if $H$ is not associated, and we obtain max-infinite divisibility by construction. Moreover, simulation from $G_{c,H}$ is straightforward provided we can simulate from $H$, while simulation from $H^t$ with non-integer $t$ may be challenging.}

In the spatial context, the above discussion generalizes to max-id processes constructed as the pointwise maximum $Z(\bm s)=\max\{X_1(\bm s),\ldots,X_N(\bm s),l(\bm s)\}$, where $X_1(\bm s),\ldots,X_N(\bm s)$ are independent realizations of $X(\bm s)$ (conditionally on $N$), and $l(\bm s)$ is their lower bound function. Relevant choices for $X$ include Gaussian processes or Student-$t$ processes, for which efficient implementations of routines to compute multivariate distribution functions exist,  or more generally elliptical processes. When $X$ is Gaussian, then $Z$  is asymptotically independent, and when $X$ is Student-$t$ with $\alpha>0$ degrees of freedom, then $Z$ is asymptotically dependent with the max-stable extremal-$t$ limit process \citep{op2013}. The top row of Figure~\ref{fig:ExtrCoefsModel} illustrates bivariate level-dependent extremal coefficients $\theta_2(u)$ defined in \eqref{eq:depcoef} for model \eqref{eq:finitemeasuremodel} when $H$ is chosen as the bivariate Gaussian or Student-$t$ distribution with $\alpha=2$ and $5$ degrees of freedom. We see that the parameter $c$ drives the steepness and the average value of the curves, with lower $c$ yielding lower curves (i.e., stronger dependence) but with steeper increase. The asymptotically dependent models based on the Student-$t$ distribution converge to their extremal-$t$ limit as $c$ increases to infinity, which explains why the curves flatten when $c$ increases. For comparison, the last row of Figure~\ref{fig:ExtrCoefsModel} shows extremal coefficients of three reference models, chosen as the max-id model of \citet{Padoan.2013} based on a limit of Gaussian process ratios, and the classical Gaussian and Student $t$ copulas. The gain in flexibility of our proposed max-id models is apparent, especially when $c$ is moderately large. We further point out that the model of \citet{Padoan.2013} is characterized by a very steep slope in the level-dependent extremal coefficients. This implies that this model has a very fast rate of convergence towards asymptotic independence, while our proposed models can better regulate this rate through the value of $c$. 

\begin{figure}[t!]
	\centering
	\includegraphics[width=0.99\linewidth]{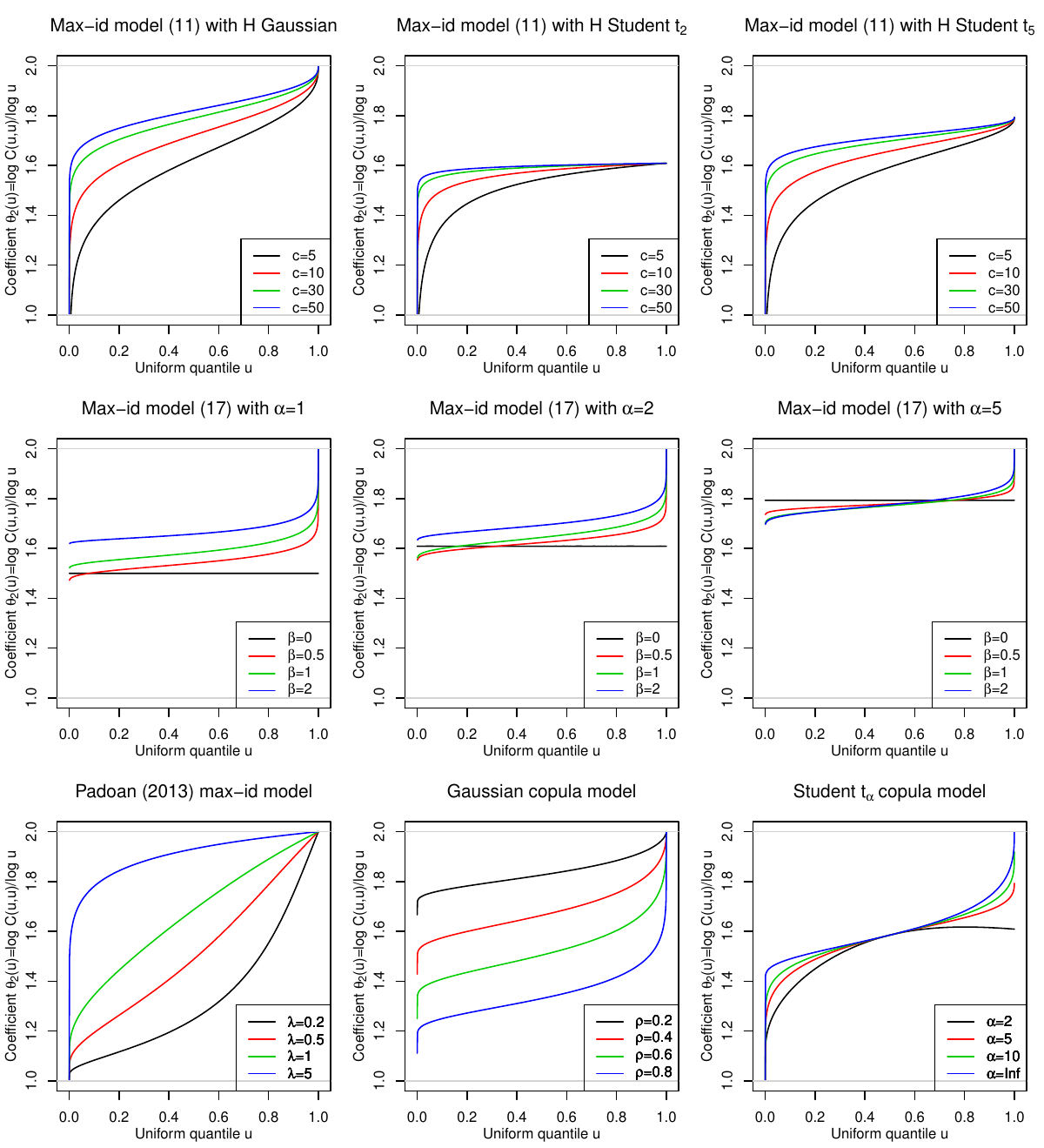}
	\caption{Bivariate level-dependent extremal coefficient $\theta_2(u)=\log C(u,u)/\log u$ in \eqref{eq:depcoef} for a a selection of models. Top row: finite exponent measure model \eqref{eq:finitemeasuremodel} for $c=5$ (black), $10$ (red), $30$ (green) and $50$ (blue), with the distribution $H$ taken as the Gaussian and Student $t$ distribution with $3$ and $10$ degrees of freedom (left to right). Middle row: infinite exponent measure model~\eqref{eq:model2} with  $\alpha=1,2,5$ (left to right) and $\beta=0$ (black), $\beta=0.5$ (red), $\beta=1$ (green) and $\beta=2$ (blue). Bottom row: model of \citet{Padoan.2013} for different variogram values $\lambda$ (left), Gaussian copula model for different correlation values (middle), Student $t$ copula model for different values of the degrees of freedom parameter (right). The models in the top and middle rows and the Student $t$ copula models are based on an underlying standard Gaussian vector $\{W(\bm s_1),W(\bm s_2)\}^{\rm T}$ with correlation $\rho(h)=0.5$.}
	\label{fig:ExtrCoefsModel}
\end{figure}

\subsection{Generalized spectral construction}
\label{sec:infinitemeasure}
\label{sec:spectralrep}
To prevent the singularity at the lower endpoint $\bm l$ (recall Section \ref{sec:finitemeasure}), we develop a general approach for constructing max-id models with infinite exponent measure by mimicking the spectral representation of max-stable processes in \eqref{eq:ms}, using a more flexible Poisson point process intensity for $\{R_i\}>0$.  We then propose parametric models that can smoothly bridge asymptotic dependence and independence. As in \eqref{eq:ms}, let $W_i(\bm s)$ be independent copies of a random process $W(\bm s)$ with $0<\E[\max\{W(\bm s),0\}]<\infty$, independent of $\{R_i\}$.   Now, instead of taking $\kappa([r,\infty))=1/r$, $r>0$, as mean measure for $\{R_i\}$, we consider max-id processes constructed as 
\begin{equation}\label{eq:msgen}
Z(\bm s)=\max_{i=1,2,\ldots} R_i W_i(\bm s), \qquad \bm s\in\calS\subset\mathbb{R}^d,\qquad 0< \{R_i\}\sim \textsc{PPP}(\kappa_{\bm \gamma}),
\end{equation}
where the mean measure $\kappa_{\bm \gamma}$, parametrized by the vector ${\bm \gamma}\in\Gamma\subset\mathbb{R}^q$, is such that $\kappa_{\bm \gamma}([0,\infty))=\infty$ but $\kappa_{\bm \gamma}([r,\infty))<\infty$ for any $r>0$. Moreover, we specify $\kappa_{\bm \gamma}$ in order to recover max-stable models as a special case of \eqref{eq:msgen}. As negative values of $W_i(\bm s)$ do not contribute to the maximum $Z(\bm s)$ in \eqref{eq:msgen}, we may replace $W_i(\bm s)$ by $\max\{W_i(\bm s),0\}$ and set $l(\bm s)=0$. The finite-dimensional exponent measure $\Lambda$ in \eqref{eq:cdf}  resulting from \eqref{eq:msgen} can be characterized by the following formula,
\begin{equation}\label{RWRadon}
\Lambda\left([\bm 0,\bm z]^{\rm c}\right)=\int_0^\infty \{1-F_{W}(\bm z/r)\}\,\kappa_{\bm \gamma}({\rm d}r)<\infty,\qquad \bm z\in(0,\infty)^D,
\end{equation} 

% leave the preceding line break (otherwise parts of the line following the equation are hidden in the pdf (in my case -- Thomas))
\noindent where $F_{W}$ denotes the distribution of the process $W(\bm s)$ observed at any finite collection of $D\geq1$ sites $\bm s_1,\ldots,\bm s_D\in\calS$.  Since it must be Radon on $E\setminus\{\bm l\}$, the values $\Lambda\left([\bm 0,\bm z]^{\rm c}\right)$ must be finite for any  $\bm z\in(0,\infty)^D$. 

An intuitive and common interpretation of \eqref{eq:msgen} is to view the max-id process $Z(\bm s)$ as the pointwise maximum of an infinite number of independent ``storms'' $R_iW_i(\bm s)$ characterized by their amplitude $R_i$ and their spatial extent $W_i(\bm s)$. In practice, we can take $W_i(\bm s)$ as one of the popular choices for modeling original event data, such as Gaussian processes, and then the scaling variables $R_i$ adjust this baseline model for accommodating more specific joint tail decay rates, which may not be well captured by the baseline.

Apart from the different measure $\kappa_{\bm \gamma}$, a major distinction between the max-stable and max-id constructions in \eqref{eq:ms} and \eqref{eq:msgen}, respectively, is that the assumption of independence between $R_i$ and $W_i(\bm s)$ is essential in \eqref{eq:ms} while it is not critical in \eqref{eq:msgen}. For example, we could choose $W_i(\bm s)$ as a Gaussian process with weakening correlation as the points $R_i$ become larger. We do not pursue this route further in this paper, but leave it for future research; rather, we here focus on choices of  $\kappa_{\bm \gamma}$ which already lead to a rich class of models.

The power-law tail of the measure $\kappa([r,\infty])=1/r$, $r>0$, in the max-stable construction \eqref{eq:ms} yields asymptotic dependence. To extend this to asymptotic independence, we propose several lighter-tailed models, with a Pareto tail on the boundary of the parameter space. Similarly to \citet{Huser.etal:2017}, our max-id construction shifts focus towards asymptotic independence while keeping the max-stable spectral representation \eqref{eq:ms} as a special case. We say that a measure $\kappa$ is Weibull-tailed if
\begin{equation}\label{eq:AIV}
\kappa([r,\infty))\sim c r^{\tau} \exp(-\alpha r^\beta), \qquad r\to\infty,
\end{equation}
for some constants $c>0$, $\alpha>0$, $\beta>0$ and $\tau\in \mathbb{R}$, where we refer to $\beta$ as the Weibull coefficient of $\kappa$. We propose the following two models for the measure $\kappa_{\bm \gamma}$ in \eqref{eq:msgen}:
\begin{align}
\kappa_{\bm \gamma}^{[1]}([r,\infty))&=r^{-(1-\alpha)}\exp\{-\alpha(r^\beta-1)/\beta\}, &&r >0,{\bm \gamma}=(\alpha,\beta)^{\rm T}\in [0,1)\times[0,\infty),\label{eq:model1}\\
\kappa_{\bm \gamma}^{[2]}([r,\infty))&=r^{-\beta}\exp\{-\alpha(r^\beta-1)/\beta\}, &&r >0,{\bm \gamma}=(\alpha,\beta)^{\rm T}\in (0,\infty)\times[0,\infty).\label{eq:model2}
\end{align}
For $\beta=0$, we interpret $\kappa_{\bm \gamma}^{[1]}$ and $\kappa_{\bm \gamma}^{[2]}$ as the limits as $\beta\downarrow 0$, giving $\kappa_{\bm \gamma}^{[1]}([r,\infty))=r^{-1}$ and $\kappa_{\bm \gamma}^{[2]}([r,\infty))=r^{-\alpha}$, $r>0$. For each model $k=1,2$, $\kappa_{\bm \gamma}^{[k]}$ is a well-defined measure that is Weibull-tailed when $\beta>0$ and that ensures $\kappa_{\bm \gamma}^{[k]}([0,\infty))=\infty$, leading to an infinite number of Poisson points in \eqref{eq:msgen}. 
With $\kappa_{\bm \gamma}^{[1]}$, we retrieve the max-stable construction \eqref{eq:ms} with unit Fr\'echet margins when $\alpha=0$ or when $\beta=0$, provided $\E[\max\{W(\bm s),0\}]=1$. With $\kappa_{\bm \gamma}^{[2]}$, we also get a max-stable model when $\beta=0$, albeit possessing $\alpha$-Fr\'echet marginal distributions.

More specifically, when the process $W(\bm s)$ in \eqref{eq:msgen} is chosen to be a standard Gaussian process, the resulting exponent measures $\Lambda_{\bm \gamma}^{[k]}$ ($k=1,2$) are well-defined Radon measures without any mass at points with component $+\infty$; see Proposition~1 in the Supplementary Material. Then, the max-stable extremal-$t$ process with $\alpha>0$ degrees of freedom arises from $\kappa_{\bm \gamma}^{[2]}$ when $\beta=0$ \citep{op2013}. {Therefore, similarly to the finite measure model in \eqref{eq:finitemeasuremodel} based on the Student-$t$ distribution $H$, we here remain in the ``neighborhood'' of the popular extremal-$t$ max-stable model, which is obtained as a special case on the boundary of the parameter space. Furthermore, by} fixing $\alpha=1$ in $\kappa_{\bm \gamma}^{[2]}$, we obtain a more parsimonious model that corresponds to the max-stable Schlather process \citep{Schlather:2002} when $\beta=0$:
\begin{equation}\label{eq:model3}
\kappa_{\bm \gamma}^{[3]}([r,\infty))=r^{-\beta}\exp\{-(r^\beta-1)/\beta\}, \quad r>0,{\bm \gamma}=\beta\in [0,\infty).
\end{equation}

For the max-stable submodels stemming from \eqref{eq:model1}, \eqref{eq:model2} and \eqref{eq:model3}, we get asymptotic dependence except in the degenerate case of complete independence. In all other non-max-stable cases, the tail decay of $\kappa_{\bm \gamma}^{[k]}$ $(k=1,2,3)$ is of Weibull type and yields asymptotic independence with Gaussian $W(\bm s)$; see Proposition~\ref{prop:chibar} in the Appendix. Under this setting, more information is carried through the coefficient of tail dependence $\eta\in(0,1]$ \citep{Ledford.Tawn.1996}, which controls the rate of convergence towards asymptotic independence; its definition is recalled in the Supplementary Material. For our proposed models $\kappa_{\bm \gamma}^{[k]}$ $(k=1,2,3)$, used in \eqref{eq:msgen} with a standard Gaussian process $W(\bm s)$ with correlation function $\rho(h)$, the coefficient of tail dependence between two sites $\bm s_1,\bm s_2$ at distance $h=\|\bm s_1-\bm s_2\|$ is $\eta(h) = \left[\{1+\rho(h)\}/ 2\right]^{\beta/(\beta+2)}$; see Proposition \ref{prop:chibar} in the appendix. The parameter $\beta$ plays a crucial role for the joint tail decay rate, while the parameter $\alpha$ also impacts the dependence structure of $Z(\bm s)$ for both $\kappa_{\bm \gamma}^{[1]}$ and $\kappa_{\bm \gamma}^{[2]}$ but to a milder degree. To illustrate the flexibility of Model~\eqref{eq:model2}, the middle row of Figure~\ref{fig:ExtrCoefsModel} displays its bivariate level-dependent extremal coefficient $\theta_2(u)$, recall \eqref{eq:depcoef}, for various parameter values. For $\beta=0$, the model is max-stable and {$\theta_2(u)\equiv\theta_2$} is constant with respect to the quantile level $u$, whereas for $\beta>0$, the dependence strength weakens as the level $u$ increases (i.e., $\theta_2(u)$ approaches $2$ as $u\to1$). Hence, the parameter $\beta$ controls ``how far'' the model is from its max-stable counterpart. The parameter $\alpha$ modulates the overall dependence strength.

Using the spectral construction \eqref{eq:msgen}, simulation mechanisms for max-id models are similar to those for max-stable models; see Section~\ref{sec:sim} for details. Our Gaussian-based models can be simulated exactly by exploiting multivariate elliptical representations.

%%%%%%%%%%%%%%%%%%%%%%%%%%%%%
%%%%%% SECTION 4: INFERENCE %%%%%%%%%
%%%%%%%%%%%%%%%%%%%%%%%%%%%%%

\section{Statistical inference}\label{sec:inference}

\subsection{Marginal standardization to uniform scale}\label{standardization}
For simplicity, we describe below our proposed estimation approach using a pairwise likelihood for data provided on uniform margins, such that we estimate only the dependence parameters. This assumes that suitable marginal distributions have been fitted first, and that the data have been pre-transformed to the uniform ${\rm Unif}(0,1)$ marginal scale prior to estimation.%, and that the expression of the exponent function $V$ is available for uniform marginal distributions. 

In practice, assume that the spatial process of maxima $\tilde{Z}(\bm s)$, $\bm s\in\mathcal S\subset\mathbb R^d$, is observed at $D$ sites $\bm s_1,\ldots,\bm s_D$. The generalized extreme-value (GEV) distribution $\tilde{G}(z)={\rm GEV}(z;\mu_j,\sigma_j,\xi_j)$ with location $\mu_j\in\mathbb R$, scale $\sigma_j>0$ and shape $\xi_j\in\mathbb R$ ($j=1,\ldots,D$) may be chosen for each margin, and then the data can be transformed to the uniform scale as $U_1=\tilde{G}(\tilde{Z}_1),\ldots,U_D=\tilde{G}(\tilde{Z}_D)$, where $\tilde{Z}_j=\tilde{Z}(\bm s_j)$ ($j=1,\ldots,D$), by plugging in marginal parameter estimates. After standardizing the data to the uniform scale, we proceed by fitting the max-id copula $C$ defined in \eqref{eq:copula}, with exponent measure $\Lambda^\star$ and exponent function given as
$$V^\star(\bm u) = \Lambda^\star\left([\bm 0,\bm u]^{\rm c}\right) = V\left\{G_1^{-1}\left(u_1\right),\ldots,G_D^{-1}\left(u_D\right) \right\},$$
where $V$ and $G_j$ ($j=1,\ldots,D$) respectively denote the exponent function and the margins of the max-id distribution \eqref{eq:cdf} for any of the models defined in Section~\ref{sec:model}.

%If the distribution functions in the model are originally not given on uniform marginal scale, we can proceed as follows; recall Equation~\ref{eq:copula}. Denote by $V$ the exponent function associated to $D$ sites $\bm s_1,\ldots,\bm s_D$, and write $G_j$, $j=1,\ldots,D$, for the corresponding marginal distributions of the max-id model. Then the exponent function $V^\star$ corresponding to the uniform marginal scale is obtained as follows for $u_j\in (0,1)$, $j=1,\ldots, D$:
%\begin{equation*}
%V^\star(\bm u) = V\left(G_1^{-1}\left(u_1\right),\ldots,G_D^{-1}\left(u_D\right) \right).
%\end{equation*}
%[SHORTLY DISCUSS JOINT ESTIMATION OF MARGINS AND DEPENDENCE? MAYBE IN THE CONCLUSION SECTION?]

\subsection{Pairwise likelihood approach}\label{sec:pairwise}
Suppose that $n$ independent replicates of a max-id process $U(\bm s)$ with standard uniform margins, parametrized by a vector $\bm \psi\in \Psi\subset\mathbb{R}^p$, are observed at $D$ sites $\bm s_1,\ldots,\bm s_D\in\calS\subset\mathbb R^d$. We write $\bm u_i=(u_{i1},\ldots,u_{iD})^{\rm T}$, where $u_{ij}$ is the $i$-th observation at the $j$-th site $(i=1,\ldots,n, j=1,\ldots,D)$. Furthermore, suppose that the density of $U(\bm s)$ with respect to Lebesgue measure on $\mathbb{R}^D$ exists. This holds for the construction \eqref{eq:msgen} (with infinite measure $\kappa_{\bm \gamma}([0,\infty))=\infty$) when the vector $\bm W=\{W(\bm s_1),\ldots,W(\bm s_D)\}^{\rm T}$ has a density $f_{\bm W}(\bm w)$ such that the intensity of the Poisson point process $\{R_i\bm W_i\}$ is obtained as $\lambda(\bm z)=\int_0^\infty   f_{\bm W}(\bm z/r)r^{-D}\kappa_{\bm \gamma}(\mathrm{d}r)$, here expressed on the original scale before transformation to uniform margins. 
The finite measure models proposed in Section \ref{sec:finitemeasure} do not admit a density because of the singularity at the lower boundary $\bm l$, but a practical workaround for this problem is to restrict the range of considered values for the parameter $c$, such that the singular mass becomes negligible. From \eqref{eq:copula}, the full likelihood function is
\begin{equation}\label{full.likelihood}
L(\bm \psi;\bm u_1,\ldots,\bm u_n)=\prod_{i=1}^n\bigg[\exp\left\{-V^\star(\bm u_i)\right\}\sum_{\pi\in\calP_D}\prod_{k=1}^{|\pi|}\left\{-V^\star_{\pi_k}(\bm u_i)\right\}\bigg],
\end{equation}
where $\calP_D$ denotes the collection of all partitions $\pi=\{\pi_1,\ldots,\pi_{|\pi|}\}$ of $\calD=\{1,\ldots,D\}$ (of size $|\pi|$), and where the exponent function $V^\star(\bm u)=\Lambda^\star\left([\bm 0,\bm u]^{\rm c}\right)$ and its partial derivatives $V^\star_{\pi_k}(\bm u)=\partial^{|\pi_k|}V^\star(\bm u)/(\prod_{j\in\pi_k}\partial u_j)$ both depend on the parameters in $\bm \psi$. For large $D$, the sum in \eqref{full.likelihood} contains too many terms to be computed and the likelihood is intractable \citep{Castruccio.etal:2016}. \citet{StephensonTawn05} improved the computational and statistical efficiency by conditioning on event times. \citet{Thibaud.etal:2016} and \citet{Huser.etal:2019} showed how to perform likelihood inference for max-stable processes by integrating out event times, but these approaches remain fairly demanding in moderately high dimensions. 
Another challenge for certain models is linked to the computation of $V^\star(\bm u)$ and $V^\star_{\pi_k}(\bm u)$ themselves. 

Pairwise likelihood inference has been widely used for max-stable processes \citep{Padoan.etal.2010,Huser.Davison:2013a}, and allows us to significantly reduce the computational burden while maintaining satisfactory statistical efficiency. This approach naturally extends to max-id processes. Instead of maximizing \eqref{full.likelihood}, pairwise likelihood inference relies on 
\begin{equation}\label{pairwise.likelihood}
PL(\bm\psi;\bm u_1,\ldots,\bm u_n)=\prod_{1\leq j_1<j_2\leq D}\left[L\{\bm\psi;(u_{1j_1},u_{1j_2})^{\rm T},\ldots,(u_{nj_1},u_{nj_2})^{\rm T}\}\right]^{\omega_{j_1;j_2}},
\end{equation}
where the innermost term is the bivariate likelihood computed from \eqref{full.likelihood}, with each independent contribution given by $L(\bm\psi;(u_{ij_1},u_{ij_2})^{\rm T})=\exp\{-V^\star(u_{ij_1},u_{ij_2})\}\{V^\star_1(u_{ij_1},u_{ij_2})V^\star_2(u_{ij_1},u_{ij_2})-V^\star_{12}(u_{ij_1},u_{ij_2})\}$, and where $\omega_{j_1;j_2}\geq0$ denotes a non-negative weight attributed to the contribution of the pair $\{j_1,j_2\}$. Usually, weights are chosen to be binary, i.e., $\omega_{j_1;j_2}\in\{0,1\}$, discarding several pairs (with zero weight) to improve computations. Moreover, weights are often fixed according to distance: $\omega_{j_1;j_2}=1$ if $\|\bm s_1-\bm s_2\|<\delta$, where $\delta>0$ is a suitable cut-off distance and $\omega_{j_1;j_2}=0$ otherwise \citep{Padoan.etal.2010}, though other strategies are also possible \citep[see, e.g.,][Chapter 3]{Huser:2013}. 

\subsection{Likelihood formulas}
\label{sec:likforms}

\paragraph{Finite measure models}
In the case of the finite exponent measure models proposed in Section~\ref{sec:finitemeasure}, the max-id construction with distribution function $G_{c,H}$  has exponent measure $\Lambda=cH$, and we obtain $V(\bm z)=c\{1-H(\bm z)\}$. If the distribution $H$ has a density $h$, one has $V_{\pi_k}(\bm z)=-cH_{\pi_k}(\bm z)$, for non-empty subsets $\pi_k$ of $\mathcal D=\{1,\ldots,D\}$ of cardinality $|\pi_k|>1$, where subscripts denote partial differentiation. In particular, $V_{\calD}(\bm z)=-ch(\bm z)$. In the case where $H$ is the multivariate standard Gaussian distribution $\Phi_D(\cdot;\Sigma)$ with correlation matrix $\Sigma$ and density $\phi_{D}(\cdot;\Sigma)$, these expressions become
\begin{align}
V(\bm z)&=c\{1-\Phi_D(\bm z;\Sigma)\};\label{Vmodel1}\\
V_{\pi_k}(\bm z)&=-c\Phi_{D-|\pi_k|}(\bm z_{\pi_k^{\rm c}}-\Sigma_{\pi_k^{\rm c};\pi_k}\Sigma_{\pi_k;\pi_k}^{-1}\bm z_{\pi_k};\Sigma_{\pi_k^{\rm c};\pi_k^{\rm c}}-\Sigma_{\pi_k^{\rm c};\pi_k}\Sigma_{\pi_k;\pi_k}^{-1}\Sigma_{\pi_k;\pi_k^{\rm c}})\phi_{|\pi_k|}(\bm z_{\pi_k};\Sigma_{\pi_k;\pi_k}),\ 1\leq |\pi_k|<D;\label{Vtaumodel1}\\
V_{\calD}(\bm z)&=-c\phi_{D}(\bm z;\Sigma),\label{VDmodel1}
\end{align}
where ${\Sigma_{A;B}}$ denotes the matrix $\Sigma$ restricted to the rows in the set $A$ and columns in $B$, $\bm z_A$ is the subvector of $\bm z$ with elements indexed by $A$, and $\pi_k^{\rm c}=\calD\setminus\pi_k$ is the complement of the set $\pi_k$ in $\mathcal D$. Expressions \eqref{Vmodel1} and \eqref{Vtaumodel1} involve the multivariate Gaussian distribution in dimension $D$ and $D-|\pi_k|$, respectively, which may be approximated using quasi Monte Carlo methods. 
This remains expensive to compute when $D$ is moderately large, which  is a common issue with most popular spatial extreme-value models. Very similar formulas are obtained if $H$ is the multivariate elliptical Student $t$ distribution with $\alpha>0$ degrees of freedom \citep{Ding.2016}, yielding asymptotic dependence except for degenerate cases of the dispersion matrix, and the Gaussian case can be seen as a limit for $\alpha\to\infty$. In general,  the finiteness of $\Lambda$ entails a point mass of the max-id distribution at the lower boundary ${\bm l}$, and therefore the likelihood formula \eqref{full.likelihood} needs to be modified accordingly as $L_{H,c}({\bm\psi};\bm u_1,\ldots,\bm u_n)=\exp(-cm)\times L({\bm\psi};\{\bm u_i,i\in\calI\})$, where $\calI=\{i\in\{1,\ldots,n\}:\bm u_i>{\bm l}\}$ and $m=n-|\calI|$ is the number of vectors $\bm u_i$ equal to the lower boundary ${\bm l}$. In practice, events are typically not observed at the lower boundary, which induces bias since the likelihood wrongfully reduces to Equation~\eqref{full.likelihood} with $m=0$. To limit this nuisance, one solution consists in assuming that $c>c_0>0$, with $c_0$ large, or adopting a censored likelihood approach with a low threshold. An alternative solution would be to remove the singular mass from the model by conditioning the number of extreme events $N$ in \eqref{eq:constructmaxid} to be at least $1$. Then, the modified multivariate distribution functions \eqref{eq:cdf} of the model would be given by $G(\bm z)/\{1-\exp(-c)\}$, $\bm z\in E$. This new distribution would still be close to the original max-id distribution when $c$ is  large, but stronger differences may arise for small $c$. In our implementation, we simply ignore the point mass for simplicity but apply the mild restriction $c>1$.

\paragraph{Infinite measure models}
In the case of the infinite exponent measure models proposed in Section~\ref{sec:infinitemeasure}, we can exploit the independence of $\{R_i\}$ and $\{\bm W_i\}$. Then, we deduce that the intensity of the Poisson point process $\{\bm X_i\}=\{R_i\bm W_i\}$, stemming from \eqref{eq:msgen} when the process is observed at $D$ sites, is $\lambda(\bm z)=\int_0^\infty f_{\bm W}(\bm z/r)r^{-D}f(r){\rm d}r$, $\bm z\in\mathbb{R}^D$, where $f_{\bm W}$ denotes the density of $\bm W_i$ and $f(r)=-{\rm d}\kappa_{\bm \gamma}([r,\infty))/{\rm d}r$ is the intensity of the Poisson point process $\{R_i\}$, provided the latter exist. For the models $\kappa_{\bm \gamma}^{[1]}$, $\kappa_{\bm \gamma}^{[2]}$ and $\kappa_{\bm \gamma}^{[3]}$ defined in \eqref{eq:model1}, \eqref{eq:model2} and \eqref{eq:model3}, the intensity $f(r)$ may be expressed as $f^{[1]}(r) = \{(1-\alpha) r^{\alpha-2} + \alpha r^{\alpha+\beta-2}\}\exp\{-\alpha(r^\beta-1)/\beta\}$, $f^{[2]}(r) = (\beta r^{-\beta-1} + \alpha r^{-1}) \exp\{-\alpha(r^\beta-1)/\beta\}$, and $f^{[3]}(r) = (\beta r^{-\beta-1} + r^{-1} ) \exp\{-(r^\beta-1)/\beta\}$, $r>0$, respectively. Using $V(\bm z)=\Lambda([-\bm\infty,\bm z]^{\rm c})=\int_{[-\bm\infty,\bm z]^{\rm c}}\lambda(\bm x){\rm d}\bm x$, we then have the following relationships:
\begin{align}
V(\bm z)&=\int_0^\infty \{1-F_{\bm W}(\bm z/r)\}f(r){\rm d}r;\label{Vmodel2}\\
V_{\pi_k}(\bm z)&=-\int_0^\infty F_{\bm W;\pi_k}(\bm z/r)r^{-|\pi_k|}f(r){\rm d}r;\label{Vtaumodel2}\\
V_{\calD}(\bm z)&=-\int_0^\infty f_{\bm W}(\bm z/r)r^{-D}f(r){\rm d}r,\label{VDmodel2}
\end{align}
where $F_{\bm W}$ is the distribution of $\bm W_i$ and $F_{\bm W;\pi_k}(\bm w)=\partial^{|\pi_k|}F_{\bm W}(\bm w)/(\prod_{j\in\pi_k}\partial w_j)$ denotes its partial derivatives. If $\bm W_i$ is multivariate standard Gaussian with correlation matrix $\Sigma$, then $F_{\bm W}(\bm w)=\Phi_D(\bm w;\Sigma)$, $f_{\bm W}(\bm w)=\phi_D(\bm w;\Sigma)$ and $F_{\bm W;\pi_k}(\bm w)=\Phi_{D-|\pi_k|}(\bm w_{\pi_k^{\rm c}}-\Sigma_{\pi_k^{\rm c};\pi_k}\Sigma_{\pi_k;\pi_k}^{-1}\bm w_{\pi_k};\Sigma_{\pi_k^{\rm c};\pi_k^{\rm c}}-\Sigma_{\pi_k^{\rm c};\pi_k}\Sigma_{\pi_k;\pi_k}^{-1}\Sigma_{\pi_k;\pi_k^{\rm c}})\phi_{|\pi_k|}(\bm w_{\pi_k};\Sigma_{\pi_k;\pi_k})$ similarly to the expression in \eqref{Vtaumodel1}. 
In this case, the expressions \eqref{Vmodel2} and \eqref{Vtaumodel2} rely on the computation of the multivariate Gaussian distribution in dimension $D$ and $D-|\pi_k|$, respectively, and the unidimensional integrals are not always available in closed form, but standard numerical methods can be used to approximate them accurately.

\subsection{Asymptotic properties and uncertainty assessment}

As pairwise likelihoods are constructed from valid likelihood terms, they inherit appealing properties from the ordinary maximum likelihood estimator \citep{Varin.etal:2011}. Under mild regularity conditions, the maximum pairwise likelihood estimator $\hat{\bm\psi}$ is strongly consistent and asymptotically normal with the well-known Godambe covariance matrix (also known as the ``sandwich matrix''). For model comparisons, the scaled composite likelihood information criterion (${\rm CLIC}^\star$) may be used \citep{Padoan.etal.2010,Davison.Gholamrezaee.2012}.

Denote by $\hat{\bm \psi}$ the maximum pairwise likelihood estimator and by $\bm\psi_0\in\Psi\subset\mathbb{R}^p$ the true parameter value, assumed to be an interior point of the parameter space $\Psi$. Under mild regularity conditions \citep{Padoan.etal.2010}, and provided that $\bm\psi$ is identifiable from the bivariate densities, we have that
\begin{equation}\label{pairwise.mle}
n^{1/2}(\hat{\bm\psi}-\bm\psi_0)\Dto\calN_p(0,J^{-1}KJ^{-1}),\qquad n\to\infty,
\end{equation}
where $n^{-1}J^{-1}KJ^{-1}$ is the sandwich matrix with sensitivity matrix $J=\E\{-{\partial^2\over\partial\bm \psi\partial\bm\psi^{\rm c}}\log PL(\bm\psi_0;\bm U)\}$ and variability matrix $K=\var\{{\partial\over\partial\bm\psi}\log PL(\bm\psi_0;\bm U)\}$ defined based on the pairwise likelihood \eqref{pairwise.likelihood}, and $\bm U=\{U(\bm s_1),\ldots,U(\bm s_D)\}^{\rm T}$ denotes a generic (uniformly distributed) vector of observations from the model. In practice, we use estimators $\hat{J}$ and $\hat{K}$ of the matrices $J$ and $K$, respectively. For $\hat J$, we use the (rescaled) hessian of the negative log pairwise likelihood evaluated at its mode, and for $\hat K$, we use the sample variance-covariance matrix of the gradient of log pairwise likelihood contributions, i.e.,
\begin{equation}\label{sandwich}
\hat J = -{1\over n}{\partial^2\over\partial\bm \psi\partial\bm\psi^{\rm c}}\log PL(\hat{\bm\psi};\bm u_1,\ldots,\bm u_n),\qquad \hat K = \widehat\var\left\{{\partial\over\partial\bm\psi}\log PL(\hat{\bm\psi};\bm u_i);i=1,\ldots,n\right\}.
\end{equation}
In \eqref{sandwich}, we use numerical derivatives for both $\hat J$ and $\hat K$. Model comparison may be performed using the composite likelihood information criterion (CLIC) defined as ${\rm CLIC}=-2\,PL(\hat{\bm\psi};\bm u_1,\ldots,\bm u_n)+2\,{\rm trace}(\hat{J}^{-1}\hat{K})$. The rescaled version ${\rm CLIC}^\star$ \citep{Davison.Gholamrezaee.2012}, based on $C^{-1}\log PL(\hat{\bm\psi};\bm u_1,\ldots,\bm u_n)$ with $C=2D^{-1}\sum_{1\leq j_1<j_2\leq D}\omega_{j_1;j_2}$ with weights $\omega_{j_1;j_2}$ defined in \eqref{pairwise.likelihood}, is easier to interpret since it recognizes that all variables in Equation~\eqref{pairwise.likelihood} appear on average $C$ times more often that they should in case of independence. In particular, when $\omega_{j_1;j_2}=1$ for all $1\leq j_1<j_2\leq D$, then $C=D-1$.

%%%%%%%%%%%%%%%%%%%%%%%%%%%%%%%%%%%%
%%%%%%% SECTION 5: SIMULATION STUDY %%%%%%%%%%%
%%%%%%%%%%%%%%%%%%%%%%%%%%%%%%%%%%%%

\section{Simulation study}\label{sec:simulation}
\subsection{Simulation of max-id processes}\label{sec:sim}
Before describing our simulation study in Section~\ref{sec:experiments}, we first briefly discuss simulation algorithms for sampling from the max-id models presented in Section~\ref{sec:model}.

%Simulation of stochastic processes is important for many inference procedures, especially in cases where Monte-Carlo approximation is required to  compensate for the absence of exact analytical formulas of quantities of interest. 

In the case of finite exponent measure models in Section~\ref{sec:finitemeasure} with $\Lambda = cH$ given a probability measure $H$, simulation is straightforward provided sampling from $H$ is possible; recall the discussion in Section~\ref{sec:finitemeasure}. We now focus on the infinite measure models presented in Section~\ref{sec:infinitemeasure}.

When the infinite mean measure of the Poisson point process $\{R_i\}$ in Equation~\eqref{eq:msgen} is given by $\kappa_{\bm \gamma}[r,\infty)=1/r$, $r>0$,  yielding a max-stable process with unit Fr\'echet margins thanks to the representation given in Equation~\eqref{eq:ms}, one can simulate the Poisson process $\{R_i\}$ by setting $R_i=1/U_i$, $i=1,2,\ldots$, where $\{U_i\}$ denotes the points from a unit rate Poisson process on the positive half-line $(0,\infty)$. A well-known way to generate ordered points $0<U_1<U_2<\cdots$ from such a process is to sample a sequence $E_1,E_2,\ldots$ of unit exponential random variables and to set $U_i=\sum_{k=1}^i E_k$, $i=1,2,\ldots$. In this way, the Poisson points $R_i$ are decreasing, which can be exploited in Equation~\eqref{eq:ms}  to generate approximate simulations of max-stable processes by truncating the maximum with a predefined accuracy \citep{Schlather:2002}. Furthermore, \citet{Schlather:2002} shows that if $W(\bm s)<C<\infty$ almost surely, then only a finite (but random) number of points $R_i$ needs to be generated for {exact} simulation of $Z(\bm s)$ in Equation~\eqref{eq:ms}. More efficient exact sampling schemes for max-stable processes have recently been proposed \citep[see, e.g.,][]{Dieker.Mikosch:2015,Dombry.etal:2016,Liu.etal:2019}.

Similarly, to simulate a max-id process defined in \eqref{eq:msgen} with a general, infinite mean measure $\kappa_{\bm \gamma}$, we propose using more general parametric transformations $R_i=T_{\bm \gamma}(U_i)$ with $T_{{\bm \gamma}}:(0,\infty)\rightarrow(0,\infty)$ given as the inverse function of the tail measure $r\mapsto \kappa_{\bm \gamma}[r,\infty)$ for $r>0$. When the transformation $T_{\bm \gamma}$ is not tractable, approximate simulation of $\{R_i\}$ may be performed in a two-step procedure: first, we generate the number of points $N$ according to a Poisson distribution with mean $\kappa_{\bm \gamma}[\varepsilon,\infty)$ for small $\varepsilon>0$. Second, we generate $N$ independent points $R_1,\ldots,R_N$ with distribution $F(r)=\kappa_{\bm \gamma}[\varepsilon,r)/\kappa_{\bm \gamma}[\varepsilon,\infty)$. We then follow the same algorithm. For the max-id constructions based on a Gaussian process $W(\bm s)$ in Section~\ref{sec:infinitemeasure}, exact simulation is possible by taking advantage of their elliptical nature based on  Proposition~\ref{prop:ellrep} in the Appendix, which establishes a finite upper bound of the spectral process when considering a finite number of locations and a modified point process; see also  \citet{Thibaud.Opitz.2015} for the max-stable extremal-$t$ case.

\subsection{Numerical experiments}\label{sec:experiments}

Using our pairwise likelihood inference approach detailed in Section \ref{sec:inference}, we consider the models \eqref{eq:model2} and \eqref{eq:model3} with infinite mean measure, combined with a standard Gaussian process $W(\bm s)$ in \eqref{eq:msgen}. Evaluating and estimating the infinite mean measure and its derivatives in \eqref{Vmodel2}, \eqref{Vtaumodel2}, \eqref{VDmodel2} is numerically more challenging than the finite mean measure counterparts in \eqref{Vmodel1}, \eqref{Vtaumodel1}, \eqref{VDmodel1}, due to integration over functions involving multivariate Gaussian densities, so we just focus here on the more ``interesting'' infinite measure models. We simulate $n=50$ independent replicates of the max-id model at $D=10,15,20,30$ sites uniformly generated in $\calS=[0,1]^2$, considering $\beta=0$ (max-stable model) and $\beta=0.5,1,2$ (asymptotically independent max-id models), and taking an isotropic exponential correlation function $\rho(h)=\exp(-h/\lambda)$ with range parameter $\lambda=0.5$ for the process $W(\bm s)$. We then estimate unknown parameters $\bm\psi$ using the pairwise likelihood estimator based on \eqref{pairwise.likelihood} with binary weights and cut-off distance $\delta=0.5$. More precisely, we perform two separate maximizations of \eqref{pairwise.likelihood}, one with $\beta=0$ fixed and one with $\beta>0$ free, and then compare the maximized pairwise likelihoods to determine the optimal value of $\beta\in[0,\infty)$. 

We first fix $\alpha=1$ corresponding to the parsimonious model \eqref{eq:model3}, and simulate max-id data based on this model for $\beta=0,0.5,1,2$. We then estimate the model parameters $(\beta,\lambda)^{\rm T}$ by maximizing the pairwise log-likelihood function, and perform $1000$ Monte Carlo simulations to assess the finite-sample variability of the estimates. The estimates are reported as boxplots in Figure~\ref{fig:boxplots1}. The results suggest that the pairwise likelihood estimator performs quite well overall and improves slightly as the dimension $D$ increases.%, although for the most complex model \eqref{eq:model2} the parameters are more difficult to identify, leading to higher uncertainty.

\begin{figure}[t!]
	\centering
	\includegraphics[width=\linewidth]{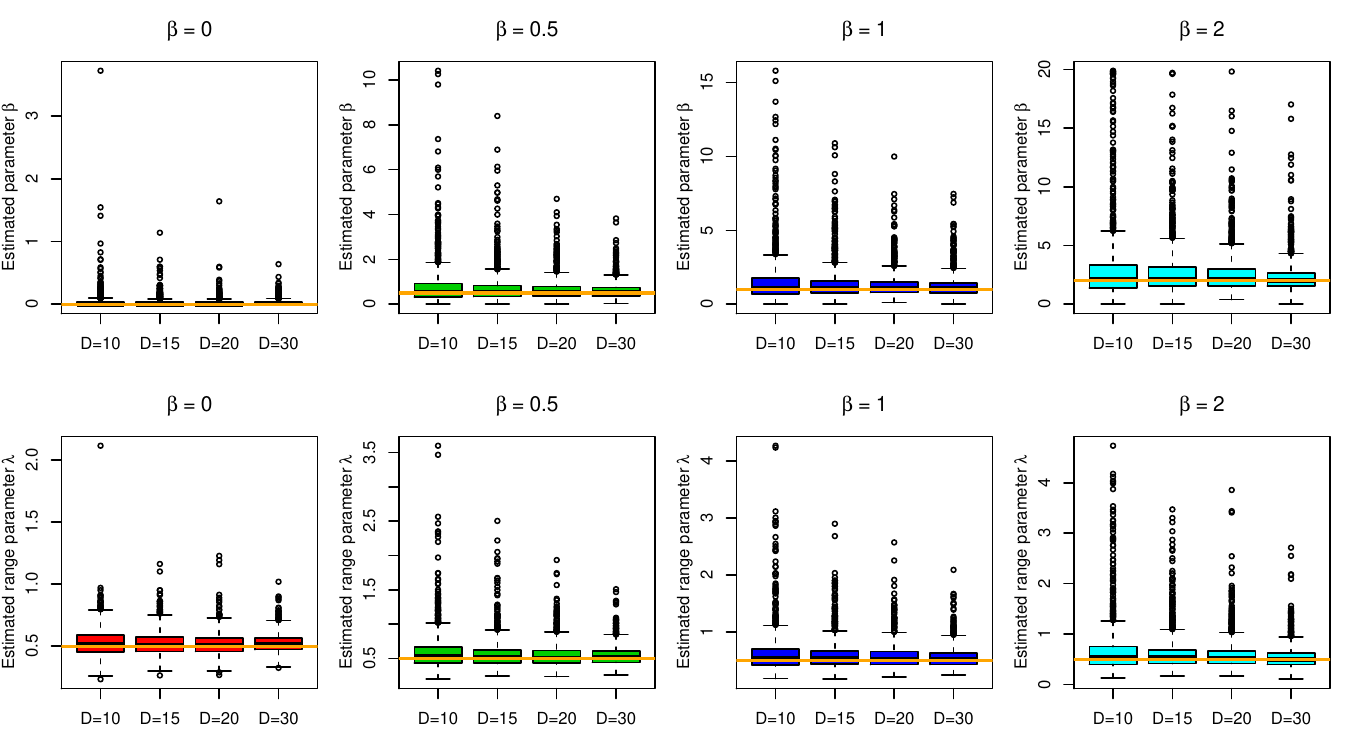}	
	%\figurebox{19.5pc}{}{}[Boxplot_sim1.eps]
	\caption{Boxplots of estimated parameters $\hat\beta$ (top) and $\hat\lambda$ (bottom) for data simulated according to the parsimonious max-id model given in \eqref{eq:model3} with $\beta=0,0.5,1,2$ (left to right) at $D=10,15,20,30$ sites randomly generated in $[0,1]^2$, with $n=50$ independent replicates. The pairwise likelihood estimator maximizing \eqref{pairwise.likelihood} with binary weights and cut-off distance $\delta=0.5$ was used. True values are indicated by orange horizontal lines. $1000$ Monte Carlo simulations were performed.}
	\label{fig:boxplots1}
\end{figure}

We then consider the general max-id model \eqref{eq:model2} at $D=20$ locations and estimate all model parameters $(\alpha,\beta,\lambda)^{\rm T}$ by pairwise likelihood. The results, shown in Figure~\ref{fig:boxplots.sim2}, suggest that even in this more complex setting, where both $\alpha$ and $\beta$ have to be jointly estimated, our pairwise likelihood estimator performs reasonably well overall, although the variability increases as $\alpha$ or $\beta$ gets larger. This is likely due to moderate identifiability issues between $\alpha$ and $\beta$ in these cases.

\begin{figure}[h!]
	\centering
	\includegraphics[width=\linewidth]{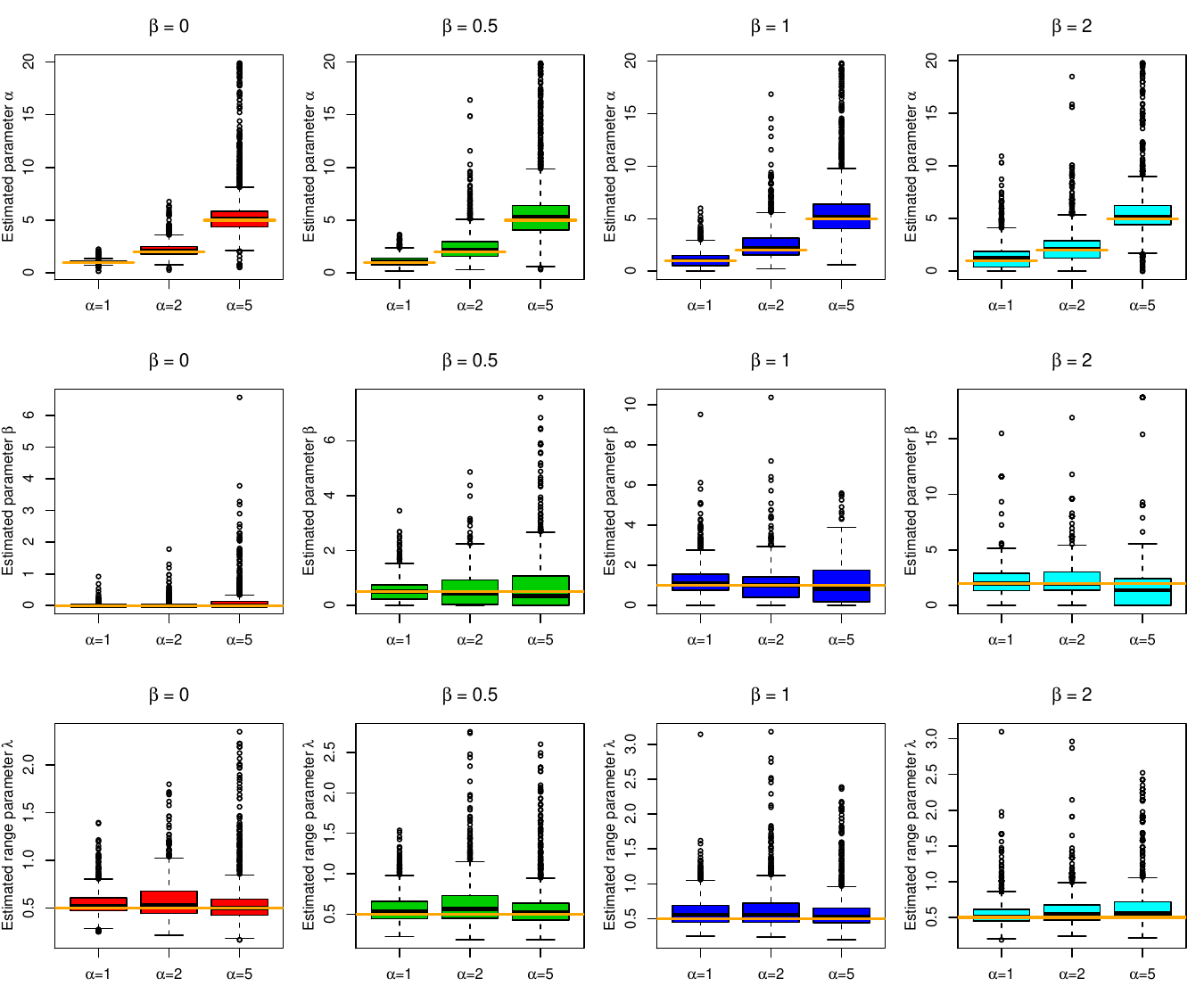}	
	%\figurebox{29pc}{}{}[Boxplot_sim2_D20.eps]
	\caption{Boxplots of estimated parameters $\hat\alpha$ (top) $\hat\beta$ (middle) and $\hat\lambda$ (bottom) for max-id data simulated according to the model given in \eqref{eq:model3} with $\alpha=1,2,5$ and $\beta=0,0.5,1,2$ (left to right panels) at $D=20$ sites randomly generated in $[0,1]^2$, with $n=50$ independent replicates. The pairwise likelihood estimator maximizing \eqref{pairwise.likelihood} where binary weights and cut-off distance $\delta=0.5$ was used. True values are indicated by orange horizontal lines. $1000$ Monte Carlo simulations were performed.}
	\label{fig:boxplots.sim2}
\end{figure}

To study the coverage probabilities of confidence intervals based on asymptotic normality and the sandwich covariance matrix \eqref{pairwise.mle} estimated using \eqref{sandwich}, we repeat our numerical experiments for the parsimonious max-id model \eqref{eq:model3}. Results are reported in the Supplementary Material. In general, coverage probabilities tend to underestimate the nominal level, which may be due to the fairly small sample size (here, $n=50$) or numerical instabilities in computing the sensitivity and variability matrices in \eqref{sandwich}. Estimating the variability matrix $K$ is indeed known to be challenging. Using a bootstrap procedure would be an obvious alternative for uncertainty assessment, but it would also be much more computationally expensive, making it impossible in our case with infinite measure models.

Finally, to assess the impact of wrongly assuming max-stability when the data are in fact max-id but asymptotically independent, we consider the same setting as above with $\alpha=1$ fixed (model \eqref{eq:model3}), $\beta=0,0.5,1,2$, exponential correlation function with range $\lambda=0.5$, $D=30$ sites in $[0,1]^2$ and $n=50$ replicates. Table~\ref{tab:PairLiks} reports the mean difference between the maximized pairwise log-likelihoods obtained with $\beta\geq0$ estimated from the data and with $\beta=0$ held fixed (Schlather max-stable model). We also report (i) the true extreme event probability $p(u)=1-{\pr\{U(\bm s_1)\leq u,\ldots,U(\bm s_{36})\leq u\}}$ that the max-id process $U(\bm s)$ (transformed to uniform margins) experiences an exceedance of the marginal level $u=0.99$ in at least one of the $36$ grid points $\bm s_1,\ldots,\bm s_{36}\in\{0,0.2,\ldots,1\}^2$, (ii) its mean estimate $\hat p_k(u)=R^{-1}\sum_{r=1}^R \hat p_{k;r}(u)$ (where $\hat p_{k;r}(u)$ is the estimate from the $r$-th simulation) based on the model with $\beta=0$ fixed ($k=1$) and $\beta\geq0$ estimated ($k=2$), and (iii) the mean relative errors $E_k=R^{-1}\sum_{r=1}^R|\hat p_{k;r}(u)- p(u)|/p(u)$ ($k=1,2$). Note that $0.01\leq p(u)\leq 1-0.99^{36}\approx 0.30$, with the lower and upper bounds corresponding to perfect dependence and independence, respectively. The results of Table~\ref{tab:PairLiks} show that incorrectly assuming max-stability leads to strongly biased joint exceedance probability estimates under asymptotic independence, while our proposed max-id models are flexible enough to give reliable results in all cases.

\begin{table}[bt]
	\centering
	\caption{Performance of max-stable and non-max-stable models constructed from \eqref{eq:msgen}, with Gaussian process $W(\bm s)$ and Poisson point process $\{R_i\}$ with measure \eqref{eq:model3} and $\beta=0,0.5,1,2$. The case $\beta=0$ corresponds to the max-stable Schlather model, while $\beta>0$ defines new models that are max-id but not max-stable. The first row (from top) reports the difference $\log\widehat{PL}_2-\log\widehat{PL}_1$ in maximized pairwise log likelihood under the model with $\beta\geq0$ estimated and with $\beta=0$ fixed, averaged over $R=1000$ simulations. The second row reports the true probability $p(u)=1-\pr\{U(\bm s_1)\leq u,\ldots,U(\bm s_{36})\leq u\}$ with $\bm s_1,\ldots,\bm s_{36}\in\{0,0.2,\ldots,1\}^2$ and $u=0.99$. The third and fourth rows report the mean estimate $\hat p_k(u)$ based on the model with $\beta=0$ fixed $(k=1)$ and $\beta\geq 0$ estimated ($k=2$). Values in parentheses are mean relative errors $E_k$ ($k=1,2$). Details are described in Section \ref{sec:experiments}.}
	\ \\ 
	\begin{threeparttable}
		\begin{tabular}{rcccc}
			\hline 
			& \thead{$\beta=0$} & \thead{$\beta=0.5$} & \thead{$\beta=1$} & \thead{$\beta=2$} \\[5pt]
			$\log\widehat{PL}_2-\log\widehat{PL}_1$ & $7.3$ & $113.0$ & $157.7$ & $178.5$ \\ 
			True tail probability $p(u)$ & $0.041$ & $0.076$ & $0.097$ & $0.122$ \\
			$\hat p_1(u)$ with $\beta=0$ fixed & $0.041$ ($3.0\%$) & $0.045$ ($40.8\%$) & $0.047$ ($51.0\%$) & $0.050$ ($59.2\%$) \\
			$\hat p_2(u)$ with $\beta\geq0$ estimated & $0.044$ ($7.4\%$) & $0.076$ ($8.6\%$) & $0.095$ ($6.3\%$) & $0.120$ ($5.5\%$)
		\end{tabular}
	\end{threeparttable}
	\label{tab:PairLiks}
\end{table}

%%%%%%%%%%%%%%%%%%%%%%%%%%%%%%%%%%%%%%%
%%%%%%% SECTION 6: REAL DATA APPLICATION %%%%%%%%%%%
%%%%%%%%%%%%%%%%%%%%%%%%%%%%%%%%%%%%%%%

\section{Analysis of Dutch wind gusts}\label{sec:application}

\subsection{Data description}

Extremes in daily wind gusts from the Netherlands (30 monitoring stations with 3241 records of daily maxima from November 11, 1999, to November 13, 2008)  were analysed by \citet{Opitz16} using an asymptotically independent Laplace random field model for high threshold exceedances. We reanalyse these data by adopting a block maximum approach. We focus on months October--March, which experience the strongest wind gusts. The maps in Figure~\ref{fig:area} show the study area with altitude, measurement sites, and also the pairs of stations that we use in the pairwise likelihood estimation procedure. Altitude variation is very slight over the Dutch territory, such that we do not consider it as a potentially informative covariate. To study wind gust extremes on various time scales, we compute daily, weekly, monthly and yearly block maxima, which yields $n_1=1594$, $n_2=220$, $n_3=52$ and $n_4=8$ maxima per site, respectively.

\begin{figure}[t!]
	\centering
	\includegraphics[width=.8\linewidth]{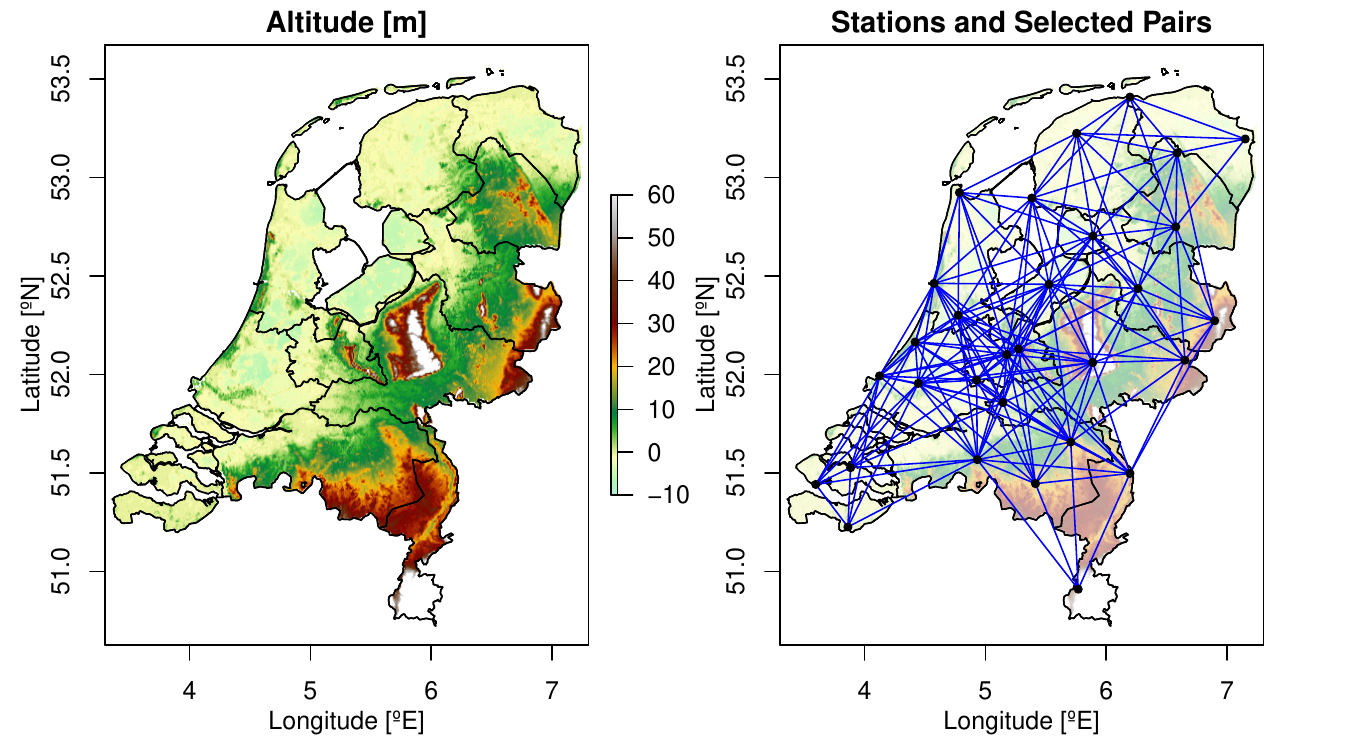}	
	\caption{Map of study area of Netherlands wind gust speed application. Left: study area with altitude. Right: Observation sites and pairs selected for pairwise likelihood inference.}
	\label{fig:area}
\end{figure}

\subsection{Marginal modeling}
We model marginal distributions separately at each location, but jointly across time scales to borrow strength across time series when few observations are available. Specifically, let $\tilde z_{ij;k}$ denote the $i$-th observation at the $j$-th monitoring station for the $k$-th time scale. We assume that the daily maxima, $\tilde z_{i_1j;1}$, follow a generalized extreme-value (GEV) distribution $\tilde G_{j;1}(z)$ with location, scale and shape parameters $\mu_j\in\mathbb{R}$, $\sigma_j>0$ and $\xi_j\in\mathbb{R}$, respectively, and that maxima for larger time scales, $\tilde z_{i_kj;k}$ ($k=2,3,4$), are also GEV-distributed according to
$$\tilde G_{j;k}(z)=\tilde G_{j;1}(z)^{b_k\theta_j}=\exp\left\{-\left(1+\xi_j{z-[\mu_j-\sigma_j\{1-(b_k\theta_j)^{\xi_j}\}/\xi_j]\over\sigma_j(b_k\theta_j)^{\xi_j}}\right)_+^{-1/\xi_j}\right\},\qquad k=2,3,4,$$
where $a_+=\max(a,0)$, $b_2=7$, $b_3=30$ and $b_4=182$ are (approximate) block sizes for weekly, monthly and yearly data, respectively, and $\theta_j\in(0,1]$ is the \emph{extremal index} specific to each station, representing the proportion of independent extremes within each block. This univariate model utilizes the common summary of temporal extremal dependence without the need to specify a full multivariate distribution for the daily observations in a block. For each site $j$, we then maximize a composite likelihood constructed by multiplying the univariate likelihood contributions of all maximum values $\tilde z_{i_kj;k}$, $i_k=1,\ldots,n_k$, $k=1,2,3,4$. 
Figure~\ref{fig:marginswinds} displays histograms of the four estimated parameters $(\hat\mu_j,\hat\sigma_j,\hat\xi_j,\hat\theta_j)^{\rm T}$ for all sites.
\begin{figure}[t!]
	\centering
	\includegraphics[width=\linewidth]{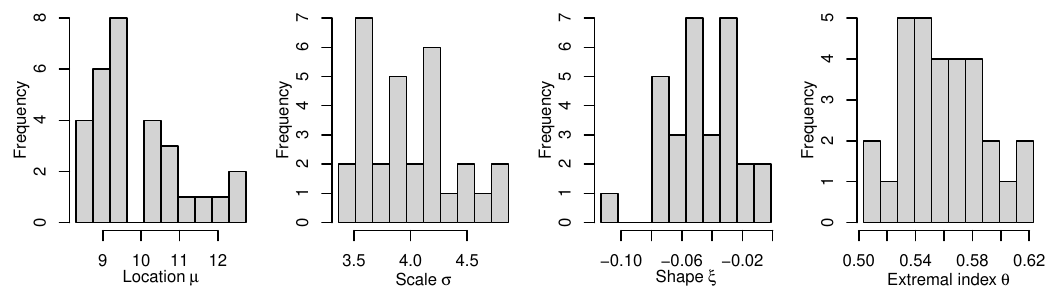}	
	%\figurebox{10.2pc}{}{}[MarginsWinds.eps]
	\caption{Histograms of estimated marginal GEV parameters for all monitoring stations.}
	\label{fig:marginswinds}
\end{figure}
In particular, the estimated shape parameters are all negative, suggesting short bounded tails, and the extremal index roughly lies in the interval $[0.5,0.6]$, revealing some mild extremal dependence in the daily time series. Quantile-quantile plots (shown in the Supplementary Material) suggest that the fits are good overall at all sites and time scales. In the following analyses and modeling steps, we standardize the data to the ${\rm Unif}(0,1)$ scale (recall Section~\ref{standardization}) and treat the transformed margins as perfectly uniform. 

\subsection{Max-id dependence modeling}

The special dependence structure of componentwise maxima suggests that these data might be well described over space by a max-id process. Although we expect the max-stability property to be reasonable for large block sizes such as for yearly maxima, it might be dubious for small block sizes such as for daily maxima; recall Figure~\ref{fig:extcoefWinds}. Indeed, the pointwise confidence bounds of the level-dependent extremal coefficients have no overlap for the lowest and highest quantiles of weekly maxima.

To corroborate this first analysis and further explore whether or not the max-stability assumption is appropriate in our context, we have studied the distribution of the weekly and the monthly maxima, taken over all of the 30 measurement locations. We do not consider here the daily and yearly time scales, since temporal dependence is still moderately strong at the daily scale, while any test procedure would lack power at the yearly scale because of the small sample size. 
%We here use a simpler approach to marginal modeling (i.e., without linking the distribution of maxima over different block sizes) to avoid any biases in this analysis. 
We have first transformed the marginal distributions to the unit Fr\'echet scale at each site separately using the empirical probability integral transform. Next, we have extracted the maximum taken over all  $D=30$ locations. Under max-stability, such spatial $D$-variate maxima must follow a Fr\'echet distribution with scale parameter equal to the corresponding extremal coefficient for the $D$ locations. Using a maximum likelihood estimate of this coefficient, we have performed a Kolmogorov-Smirnov test to verify whether the fitted Fr\'echet distribution (obtained under max-stability) has a good match with the corresponding empirical distribution. Resulting $p$-values are $0.03$ for weekly data, and $0.18$ for monthly data. This again casts strong doubt on the max-stability assumption for weekly data, while the situation for monthly maxima (with a moderate number of $52$ observations) is less clear.

Treating the estimated margins as exact, we then fit several max-id models with finite and infinite exponent measure through the pairwise likelihood approach. For the infinite exponent measure models with structure given by   \eqref{eq:msgen}, we use an isotropic Gaussian process $W(\bm s)$ with powered exponential correlation function $\rho(h)=\exp\{-(h/\lambda)^\nu\}$, $h\geq 0$, with range $\lambda>0$ and smoothness $\nu\in(0,2]$, and we use the Poisson point process mean measure proposed in \eqref{eq:model2}, which depends on the parameters $\alpha>0$ and $\beta\geq0$. We consider four infinite measure max-id models, fitted separately for each time scale: $\alpha=1$ and $\beta=0$ both fixed (Schlather max-stable model); $\alpha>0$ free and $\beta=0$ fixed (extremal-$t$ max-stable model with $\alpha$ degrees of freedom); $\alpha=1$ fixed and $\beta\geq0$ free (parsimonious model \eqref{eq:model3}); $\alpha>0$ and $\beta\geq0$ both free (general unconstrained max-id model \eqref{eq:model2}). 
For finite measure models, we consider the construction \eqref{eq:finitemeasuremodel} from Section~\ref{sec:finitemeasure} with exponent measure $cH$, $c>0$, where $H$ is defined as the standard Gaussian or Student-$t$ distribution with the powered exponential correlation function as defined above. Additionally, we also consider three reference models: the max-id process of \citet{Padoan.2013} with two parameters (range $\lambda>0$ and smoothness $\nu\in(0,2]$ in the variogram model $(h/\lambda)^{\nu}$), as well as two classical copula models corresponding to Gaussian and Student-$t$ processes with a parametrization comparable as above, which do not possess specific properties for modeling maxima. 

All models were estimated by maximizing the pairwise likelihood \eqref{pairwise.likelihood}, considering all pairs of locations less than $\delta=100$km apart (i.e., keeping roughly $40\%$ of possible pairs). Selected pairs contributing to the pairwise likelihood formula are shown in the graph on the right-hand display of  Figure~\ref{fig:area}. For the most complex max-id model with infinite exponent measure, a single fit took about 30min, 3.5h, 15h, and 4 days for yearly, monthly, weekly and daily maxima, respectively, on a workstation with 20 cores exploited for computing the pairwise likelihood in parallel. Fitting the models with finite exponent measure is much faster since the bivariate Gaussian or Student-$t$ distribution functions do not appear in functions to be integrated over. Table~\ref{tab:resultswinds} reports the parameter estimates for all fitted models, as well as the CLIC$^\star$ values (with matrices in \eqref{sandwich} adjusted based on the effective sample size to account for the moderate temporal dependence at the daily scale). %We also fitted the max-id model of \citep{Padoan.2013}, based on a limit for Gaussian process ratios, through pairwise likelihood to estimate its two parameters (range and shape).
%We  found that  this model has  by far the lowest CLIC$^\star$ value among all models, which may be due to the relatively fast joint tail decay rates in this model, inappropriate for our windgust data.

\begin{table}[bt]
{\small 
	\centering
	\caption{Summary of fitted models for Dutch wind gust maxima. Rows correspond to different time scales. Panels are numbered from top to bottom in the following description. 
		Results are given for finite exponent measure model using Gaussian (top panel, left) and Student $t$  (top panel, right)	distribution for  $H$, with estimated parameters $\alpha$ (degrees of freedom, fixed to $\infty$ for the Gaussian max-id model), $c$ (effective block size), $\log(\lambda)$ (log range) and $\nu$ (smoothness), and for infinite exponent measure  models built from \eqref{eq:msgen} and \eqref{eq:model2}, with parameters $\alpha$, $\beta$, $\log(\lambda)$ and $\nu$  estimated for the Schlather max-stable model with $\alpha=1$, $\beta=0$ (second panel, left),  the extremal-$t$ max-stable model with $\beta=0$ (second panel, right), the parsimonious max-id model \eqref{eq:model3} with $\alpha=1$ (third panel, left), and the unconstrained max-id model \eqref{eq:model2} (third panel, right). The bottom panel reports results for Gaussian and Student-$t$ copula models (bottom panel, left and middle), with parameter notation as for their maxi-id counterparts in the top panel, and for the \citet{Padoan.2013} model (bottom panel, right) with range and shape parameter.   Model fitting was based on \eqref{pairwise.likelihood} with binary weights and pairs less than $100$km apart. The scaled composite likelihood information criterion (CLIC$^\star$) is also reported. For a fixed dataset, a lower  CLIC$^\star$ indicates a better model. For each time scale, the best model (lowest CLIC$^\star$) is indicated in bold font.}
	\ \\
	
	\begin{threeparttable}
		\begin{tabular}{r rrrr rr rrrrr}
			\hline
			\multicolumn{1}{c}{} &\multicolumn{5}{c}{Gaussian max-id model}&   \ \  &	\multicolumn{5}{c}{Student $t$ max-id model}\\
			& $\alpha$ & $\hat{c}$ & $\log(\hat\lambda)$ & $\hat\nu$ & CLIC$^\star$ &  \ \ & $\hat\alpha$ & $\hat{c}$ & $\log(\hat\lambda)$ & $\hat\nu$ & CLIC$^\star$\\[5pt]
			Daily& $\infty$ &40.8&9.58&0.7&-44726.4&\  \  &7.18&11.7&8.75&0.7&$\bm{-45281.5}$\\ 
			Weekly& $\infty$ &20.4&11.56&0.43&-5438.4&\  \  &6.31&13.3&8.4&0.65& $\bm{-5494.2}$ \\ 
			Monthly& $\infty$ &32.7&10.17&0.45&-797&\  \  &3.36&12.4&7.81&0.56&$\bm{-838.9}$\\ 
			Yearly& $\infty$ &7.9&13.12&0.14&-56.2&\  \  &1.01&45.5&5.1&0.66&$\bm{-73.9}$\\ 
		\end{tabular}
		\ 
		\begin{tabular}{r rrrr rr rrrrr}
			\hline
			\multicolumn{1}{c}{} &\multicolumn{5}{c}{Schlather max-stable model}&   \ \  &	\multicolumn{5}{c}{Extremal-$t$ max-stable model}\\
			& $\alpha$ & $\beta$ & $\log(\hat\lambda)$ & $\hat\nu$ & CLIC$^\star$ & \ \ &	$\hat\alpha$ & $\beta$ & $\log(\hat\lambda)$ & $\hat\nu$ & CLIC$^\star$ \\[5pt]
			Daily&1&0&8.93&1.58&-42418.9&\  \  &4.7&1&12.07&1.59&-44339.2\\ 
			Weekly&1&0&8.63&1.48&-5090.8&\  \  &4.57&1&12.16&1.49&-5366.4\\ 
			Monthly&1&0&7.02&1.52&-789.3&\  \  &3.29&1&9.69&1.53&-825.7\\ 
			Yearly&1&0&6.11&1.4&-71.1&\  \  &1.22&1&6.76&1.38&-69.2\\ 
		\end{tabular}
		\ 
		\begin{tabular}{r rrrr rr rrrrr}
			%	\multicolumn{6}{c}{}\\
			\hline
			\multicolumn{1}{c}{} & \multicolumn{5}{c}{Parsimonious max-id model~\eqref{eq:model3}}& \ \ &	\multicolumn{5}{c}{General max-id model~\eqref{eq:model2}}\\\
			& $\alpha$ & $\hat\beta$ & $\log(\hat\lambda)$ & $\hat\nu$ & CLIC$^\star$ & \ \ &	$\hat\alpha$ & $\hat\beta$ & $\log(\hat\lambda)$ & $\hat\nu$ & CLIC$^\star$ \\[5pt]
			Daily&1&1.79&10.86&1.59&-44907.8&\  \  &2.69&1.4&11.12&1.61&-45084.1\\ 
			Weekly&1&1.68&9.49&1.64&-5436.3&\  \  &2.57&1.29&11.33&1.49&-5459.8\\ 
			Monthly&1&1.37&8.8&1.53&-830.6&\  \  &2.28&0.58&8.54&1.62&-832.1\\ 
			Yearly&1&0.17&7.15&1.3&-70.8&\  \  &1.13&0.02&6.39&1.42&-71.4\\ 
		\end{tabular}%\quad\qquad
		\ 
		\begin{tabular}{r rrrr r rrrr r rrr}
			\hline
			\multicolumn{1}{c}{} &\multicolumn{4}{c}{Gaussian copula} &   \ \  &	\multicolumn{4}{c}{Student-$t$ copula}&   \ \  &	\multicolumn{3}{c}{Padoan model}\\
			& $\alpha$ & $\log(\hat\lambda)$ & $\hat\nu$ & CLIC$^\star$ &  \ \ & $\hat\alpha$ & $\log(\hat\lambda)$ & $\hat\nu$ & CLIC$^\star$ &  \ \ & $\log(\hat\lambda)$ &  $\hat \nu$ & CLIC$^\star$ \\[5pt]
			Daily&Inf&9.88&0.43&-44710.8&\  \  &9.15&8.97&0.52&-45228.6&\  \  &7.74&0.59&-28176.9\\ 
			Weekly&Inf&9.82&0.37&-5420.9&\  \  &9.61&9.32&0.4&-5483.9&\  \  &7.62&0.47&-3181\\ 
			Monthly&Inf&8.12&0.39&-784.4&\  \  &6.24&8&0.41&-818.6&\  \  &5.96&0.57&-422\\ 
			Yearly&Inf&12.75&0.09&-48.4&\  \  &3.02&7.75&0.25&-63.8&\  \  &5.1&0.29&-27.3\\ 
		\end{tabular}
	\end{threeparttable}
	\label{tab:resultswinds}
}
\end{table}

Throughout,  estimated range parameters $\hat\lambda$ are large and suggest that spatial dependence is quite strong. In contrast, the estimated smoothness parameters $\hat\nu$ are often around $0.5$, which shows that there is small-scale variability. We first compare the max-stable models with the max-id extensions based on the spectral construction.
The parameter estimates for the Schlather model suggest that max-stability might be dubious for these data: $\hat\lambda$ and $\hat\nu$ are both decreasing with larger time scales, suggesting a weakening of spatial dependence as wind gusts become more extreme. The results for the extremal-$t$ model seem to confirm this, although the parameter $\alpha$ has the opposite effect. More affirmative conclusions can be drawn by comparing the fits of the max-id models~\eqref{eq:model2} and \eqref{eq:model3} with their max-stable counterparts obtained by fixing $\beta=0$. For yearly maxima, $\hat\beta$ is fairly close to zero in both non-max-stable models. For model~\eqref{eq:model2}, the $95\%$ confidence interval for $\beta$ (not shown) includes $0$, suggesting that the max-stable assumption is reasonable in this case. Furthermore, the CLIC$^\star$ values are all very similar for yearly maxima, suggesting that the parsimonious Schlather max-stable model might be appropriate. By contrast, for daily, weekly and monthly maxima, the estimates of $\hat\beta$ in non-max-stable models are always significantly different from zero at the $95\%$ confidence level. Interestingly, for the max-id models~\eqref{eq:model2} and \eqref{eq:model3}, $\hat\beta$ decreases monotonically to zero as the block size gets larger. This implies that these block maxima tend to be closer to a max-stable process as the block size increases, while $\beta$ provides extra flexibility at sub-asymptotic regimes characterized by small block sizes. 

In comparison with the infinite measure models, the max-id models with finite exponent measure, based on the Gaussian or Student-$t$ processes in their construction, are both very competitive. In particular, the Student-$t$ max-id model shows the best performance overall, as measured in terms of CLIC$^\star$, across all models and block sizes. The parameter $c$, corresponding to the expected number of independent replicates used to compute maxima, is always estimated to lie above $7.9$, such that the singular mass arising in the density is negligible.

As for the reference models, both the classical Gaussian and Student-$t$ copula models succeed in outperforming some of the max-stable and max-id models for small block sizes only (daily, weekly), but fail for larger block sizes (monthly, yearly). Interestingly, the model of \citet{Padoan.2013} is by far and consistently the worse in terms of CLIC$^\star$, which may be due to the relatively fast joint tail decay rate in this model, inappropriate for our wind gust data.

Overall, the CLIC$^\star$ values are strongly supporting the non-max-stable models, especially for daily and weekly maxima. In particular, the Student-$t$ copula and the Student-$t$ max-id models outperform the extremal-$t$ model. Both models show decreasing strength of extremal dependence (recall Figure~\ref{fig:ExtrCoefsModel}) while tending towards the asymptotically dependent extremal-$t$ limit process, and they provide a better fit at finite quantile levels for the wind gust data. 

\begin{figure}[t!]
	\centering
	\includegraphics[width=\linewidth]{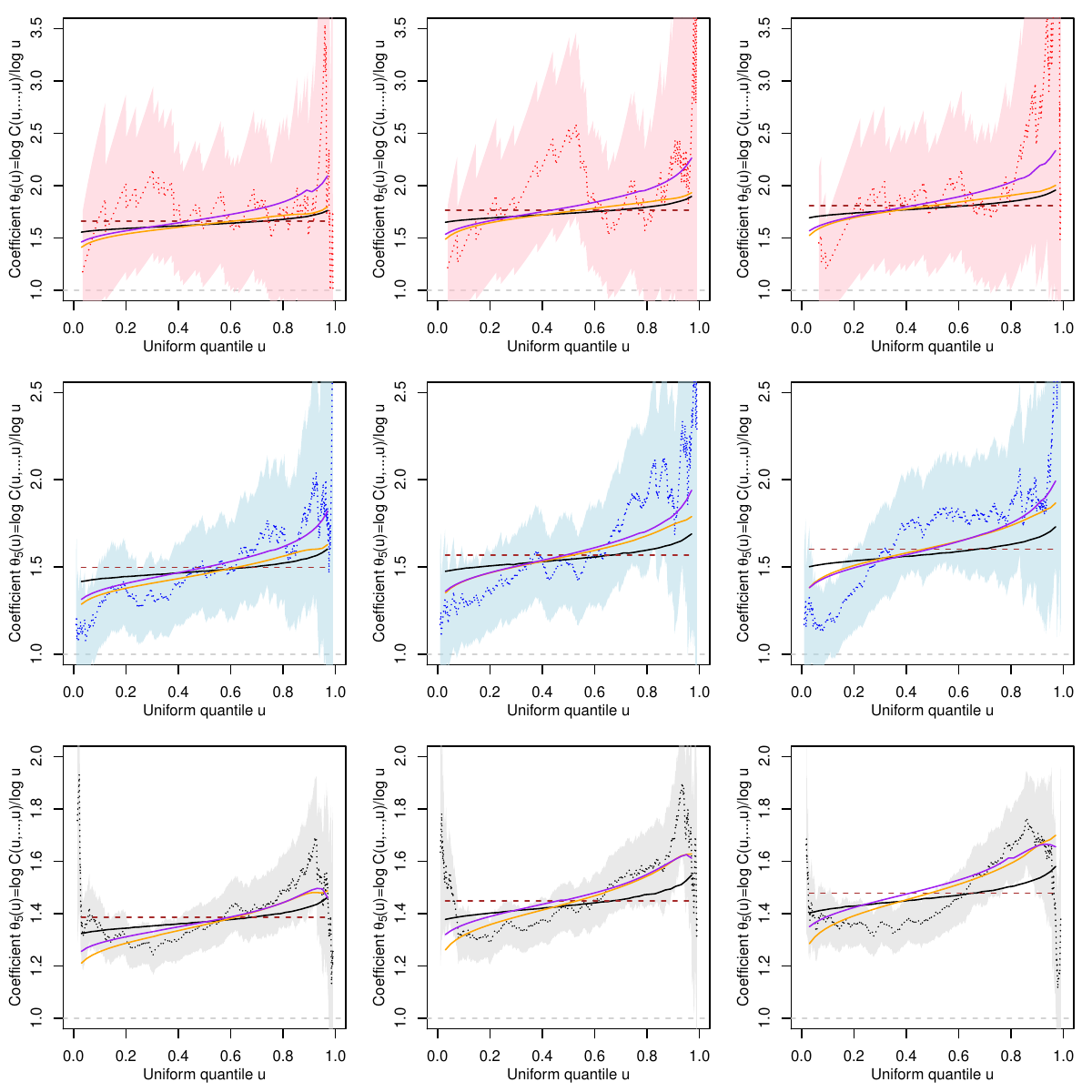}
	%\figurebox{23.5pc}{}{}[extcoef5_Wind_Winter.eps]
	\caption{Level-dependent extremal coefficients $\theta_D(u)$ defined in \eqref{eq:depcoef} for monthly (top), weekly (middle) and daily (bottom) maxima for three subsets of sites of dimension $D=5$ at average distance $33$km (left), $57$km (middle) and $74$km (right). Dotted lines are empirical estimates with $95\%$ pointwise confidence bands (shaded areas).  Solid lines correspond to fitted extremal coefficients for the Gaussian (purple) and Student-$t$ (orange) max-id finite measure models, and the infinite exponent measure max-id models \eqref{eq:model2} (black). Dashed lines report fitted extremal coefficients of the extremal-$t$ max-stable model.}
	\label{fig:extrCoef5}
\end{figure}

To visually assess the goodness-of-fit of the estimated max-id models, Figure~\ref{fig:extrCoef5} displays empirical and fitted level-dependent extremal coefficients $\theta_D(u)$ for daily, weekly, and monthly maxima, for three subsets of sites of dimension $D=5$ at various distances. We compare the fits of the max-stable extremal-$t$ model with the max-id model \eqref{eq:model2}, and with the Gaussian and Student-$t$ max-id models. Unlike the quite rigid extremal-$t$ model, our  new max-id models can capture weakening of dependence. In particular, the Gaussian and Student-$t$ max-id models have relatively steep slopes and follow the empirical curve quite well. More flexible max-id models could be designed for example by allowing for spatial anisotropy or by choosing processes $W_i(\bm s)$ in the spectral construction \eqref{eq:msgen} that depend on the Poisson points $R_i$, which might further improve results. %Our models are able to capture a general tendency of maxima (taken over several sites) to show systematically higher values than those expected with max-stability, resulting in a monotonically strictly increasing level-dependent extremal coefficient. 

To summarize, the CLIC$^\star$ values, the estimated parameters (Table~\ref{tab:resultswinds}), and the visual diagnostics strongly suggest that our proposed max-id models outperform their max-stable extremal-$t$ counterpart for small and moderate block sizes, as well as classical  copula models from geostatistics for large block sizes.

%%%%%%%%%%%%%%%%%%%%%%%%%%%
%%%%%% SECTION 6: DISCUSSION %%%%%%%
%%%%%%%%%%%%%%%%%%%%%%%%%%%

\section{Discussion}\label{sec:discussion}
We have proposed new families of sub-asymptotic spatial models for block maxima, that are more flexible than max-stable processes while remaining max-infinitely divisible. Specifically, we have developed general constructive approaches for max-id processes based on finite and infinite exponent measures, which possess appealing tail properties and remain in the ``neighborhood'' of popular max-stable models. In particular, our proposed infinite measure models extend the spectral characterization of max-stable processes, using an infinite number of independent replicates of a random process of the form $RW(\bm s)$ to build asymptotically independent max-id structures. In our examples, we took the process $W(\bm s)$ to be Gaussian and \emph{independent} of the random scale $R$ for simplicity. Flexible dependence structures may also be obtained by letting $W(\bm s)$ depend on $R$ in such a way that the correlation strength of the $W$ process decreases as $R$ increases, but likelihood calculations may be further complicated. Regarding simulation of max-id processes, the exact simulation algorithm for max-stable processes developed \citet{Dombry.etal:2016} and based on the conditional simulation of Poisson points in the case of infinite intensity measures, could be adapted. 

%\thom{[SIMILAR FOR FINITE MEASURE MODEL WITH REFERENCE $c=\infty$?]}
Another interesting research direction would be to formulate a Bayesian max-id model based on \eqref{eq:finitemeasuremodel} (or \eqref{eq:model2}), in which a prior distribution for the parameter $c>0$ (or $\beta\geq0$) would be chosen to shrink max-id models towards their ``simpler'' max-stable counterparts arising for $c\to\infty$ (or $\beta=0$). One could take advantage of the concept of penalized-complexity (PC) priors \citep{Simpson.al:2014}; see the PC prior derivation of the tail index in \citet{Opitz.al.2018} for an extreme-value example. Alternatively, we could assume that the prior for $1/c$ (or $\beta$) is a mixture between a continuous distribution on $(0,\infty)$ and a point mass at zero. In the Bayesian context, the recent work of \citet{Bopp.al.2020} proposes a hierarchical max-id model extending the max-stable \citet{Reich.Shaby:2012a} model.

\section{Supplementary Material}
Supplementary Material  is available in the Appendix section~\ref{Appendix:SM} below and includes an illustration of how negative correlation in the bivariate Gaussian distribution impedes the max-id property,  details on the coefficient of tail dependence used to summarize asymptotic independence structures, some auxiliary results for our Gaussian-based infinite mean measure models, estimated coverage probabilities for confidence intervals derived from the asymptotic behavior of the maximum pairwise likelihood estimator, and QQ-plots for the marginal model in the data application to Dutch wind gusts.

\appendix

\section{Results for the Gaussian-based constructions}
\label{sec:proofs}

\begin{proposition}[Elliptical point process representation]\label{prop:ellrep}\ 
	Consider the Poisson process with points  $\{R_iW_{s^\star,i},\ i=1,2,\ldots\}$, where $W_{s^\star,i}$ are independent copies of a standard Gaussian vector defined over a configuration of sites $s^\star=\{s_1,\ldots,s_D\}$ with correlation matrix $\Sigma_{s^\star}$, and  $\{R_i\}$ are the points of a Poisson process with intensity measure $\kappa_{\gamma}^{[k]}$ ($k=1,2,3$), given  in \eqref{eq:model1}, \eqref{eq:model2} and \eqref{eq:model3} respectively, with  independence between $W_{s^\star,i}$ and $\{R_i\}$.
	Then the  Poisson process  has elliptical representation 
	\begin{equation}\label{eq:ppell}
	\{R_iW_{s^\star,i}, i=1,2,\ldots \} \Deq \{\tilde{R}_i\Sigma_{s^\star}^{1/2}S_i,\ i=1,2,\ldots \}
	\end{equation}
	where $\Deq$ means equality in distribution, $\{\tilde{R}_i,\ i=1,2,\ldots\}\Deq\{R_iR_{W,i},\ i=1,2,\ldots\}$ is a Poisson process and $R_{W,i}>0$ are independent random variables following the chi-distribution $F_{\chi_D}$  with $D$ degrees of freedom, 
	%with density
	%$$f_{\chi_D}(x) = 2^{1-D/2} \Gamma(D/2)^{-1}x^{D-1} \exp(-x^2/2), \qquad x>0,$$
	independently of random vectors $S_i$ uniformly distributed over the unit sphere $\mathcal{S}_{D-1}$. The intensity measure $\tilde{\kappa}_{\gamma}^{[k]}$  of $\{\tilde{R}_i,\ i=1,2,\ldots\}$ is characterized through its tail measure
	\begin{equation}\label{eq:elltailmeasure}
	\tilde{\kappa}_{\gamma}^{[k]}([z,\infty))=\int_0^\infty \kappa_{\gamma}^{[k]}([z/r,\infty)) f_{\chi_D}(r)\,\mathrm{d}r=\int_0^\infty \overline{F}_{\chi_D}(z/r) f^{[k]}(r)\,\mathrm{d}r, \qquad z>0.
	\end{equation}
	%The  Poisson intensity measures $\tilde{\kappa}^{[k]}$ of  $\{R_iR_{S,i},\ i=1,2,\ldots\}$  have Weibull-type univariate tail behavior with Weibull coefficient $2\beta/(\beta+2)$. 
\end{proposition}

\begin{proof}
	A multivariate standard Gaussian random vector $W_{s^\star}$ has elliptical representation $W_{s^\star}\Deq R_{W}\Sigma_{s^\star}^{1/2} S$ where $R_{W}>0$ and $R_{W}^2$ follows a chi-squared distribution with $D$ degrees of freedom, independently of a uniform random vector $S$ over the unit sphere $\mathcal{S}_{D-1}$, and with $\Sigma_{s^\star}^{1/2}\Sigma_{s^\star}^{T/2}=\Sigma_{s^\star}$. Therefore, the Poisson process used in the max-id construction is characterized by the elliptical construction \eqref{eq:ppell}. The tail measure representations in \eqref{eq:elltailmeasure} are obtained by integrating out the distribution of the random factors $R_{W,i}$ in the Poisson points $\{R_iR_{W,i},\ i=1,2,\ldots\}$ and  integration by parts.
	%	; directly switching between the two representations is possible via an integration-by-parts argument. 
\end{proof}

\begin{proposition}[Asymptotic independence in bivariate max-id vectors]\label{prop:chibar}\ 
	Consider the bivariate max-id distribution $Z=(Z_1,Z_2)^{\rm T}=\max_{i=1,2,\ldots} R_i(W_{1,i},W_{2,i})^{\rm T}$,  constructed using independent copies $(W_{1,i},W_{2,i})^{\rm T}$ of a standard Gaussian random vector $W=(W_1,W_2)^{\rm T}$ with correlation coefficient $\rho\in[-1,1]$, independent of the points $\{R_i\}$ of a Poisson process distributed according to one of the intensity measures $\kappa_{\gamma}^{[k]}$ ($k=1,2,3$).  Then, for $\beta>0$ (and/or $\alpha>0$ for $\kappa_{\gamma}^{[1]}$), the distribution of $Z$ is asymptotically independent with coefficient of tail dependence $\eta= \left\{(1+\rho)/ 2\right\}^{\beta/(\beta+2)}$.
	%\begin{equation}
	%\eta= \left\{(1+\rho)/ 2\right\}^{\beta/(\beta+2)}.
	%\end{equation}
\end{proposition}

\begin{proof}
	Proposition \ref{prop:ellrep} provides the equivalent elliptical construction of $Z$ as $Z\Deq \{\tilde{R}_i \Sigma^{1/2} S_i,\ i=1,2,\ldots\}$, where $\Sigma$ is a 2-by-2 correlation matrix with unit diagonal entries and off-diagonal entries $\rho$, and
	%\begin{equation*}
	%Z\Deq \{\tilde{R}_i \Sigma^{1/2} S_i,\ i=1,2,\ldots\}, \qquad \Sigma=\begin{pmatrix}1 & \rho \\ \rho & 1 \end{pmatrix},
	%\end{equation*}
	the $S_i$s are independent bivariate spherical random vectors. According to Proposition \ref{prop:weibtail} in the Supplementary Material, the Poisson process $\{\tilde{R}_i\}$ is Weibull-tailed with coefficient $2\beta/(2+\beta)$. Using the elliptical structure of the bivariate measure $\Lambda$ and the tail approximation \eqref{eq:approxtail}, we can apply results on the joint tail behavior of elliptical distributions with Weibull-tailed radial variables \citep{ha2010a,Huser.etal:2017} to characterize the joint tail behavior of the distribution $G$ of $Z$, which yields the coefficient coefficient of tail dependence $\eta= \left\{(1+\rho)/ 2\right\}^{\beta/(\beta+2)}$; see \citet[Theorem 2.1]{ha2010a}.% and \citet[Theorem 2]{Huser.etal:2017}.
	%the tail measure $V(z)=\Lambda_{S}(-\infty,z]^C$ for  $z\not\leq  z_0$ with $z_0$ chosen such that  $V(z_0,\ldots,z_0)\ll 1$ corresponds to the  characterizes a centered elliptical distribution 
	% Using the max-id tail approximation \eqref{eq:approxtail} relating the tails of the multivariate max-id distribution with $H$, we conclude that the distribution of the max-id random vector constructed from the componentwise maximum
	% \begin{equation*}
	% \tilde{Z}_{S,1} = \max_{R_i >1} R_iW_{S,i}
	% \end{equation*}
	% is Weibull-type with coefficient $2\beta/(\beta+2)$. 
	%The bivariate standard Gaussian vector $W=(W_1,W_2)^{\rm T}$ has elliptic representation $R_W (U_1, \rho U_1+\sqrt{1-\rho^2}U_2)^{\rm T}$ with $R_W\indep (U_1,U_2)^{\rm T}$ and $(U_1,U_2)^{\rm T}$ distributed uniformly on the unit circle. We let parameter and variable subscripts such as in $R_W$ and $\alpha_W$ refer to the vector $W$. Since $R_W^2 \sim \chi^2_2=\Gamma(a=1,b=2)$ is gamma distributed, we get the tail expansion
	%\begin{equation*}
	%\pr(R_W> r) = \pr(R_W^2 > r^2)  \sim (r/2)^{1-1}\{\Gamma(1)\}^{-1}\exp(-r^2/2)=   \exp(-r^2/2), \quad r\rightarrow \infty.
	%\end{equation*}
	%Therefore, $R_W$ has  Weibull-type tail as in \eqref{eq:AIV} with parameters $\alpha_{W}=1, \  \beta_{W}=2, \gamma_W=0,\ \delta_{W}=1/2$ in obvious notation. 
\end{proof}
%The joint tail expression  of the distribution $G$ of the bivariate max-id random vector in \ref{prop:chibar} could be further detailed following the arguments and expressions given in the  the proofs in the Appendix and Supplementary Material of \citet{Huser.etal:2017}. 

\section{Supplementary Material}\label{Appendix:SM}

This document provides a counter-example of a joint distribution that is not max-id (Section~\ref{SMsec:maxid}), details on the coefficient of tail dependence used to summarize asymptotic independence structures (Section~\ref{sec:eta}), some auxiliary results for our Gaussian-based models (Section~\ref{SMsec:theory}), additional results on coverage probabilities for confidence intervals derived from the asymptotic normality of the maximum pairwise likelihood estimator (Section~\ref{SMsec:coverage}) and further details on the marginal fit in the data application to Dutch wind gusts (Section~\ref{SMsec:applic}). Cross-references to our main paper are written as \citet{ourpaper}. 

\subsection{Failure of max-infinite divisibility under negative association}\label{SMsec:maxid}
A simple counter-example of a distribution without the max-id property is the bivariate standard Gaussian distribution $\Phi_2(\cdot;\rho)$ with negative correlation $\rho$ \citep[Section 5.2]{Resnick.1987}. Figure~\ref{fig:gaussian} displays the ``density'' $\frac{\partial^2}{\partial z_1 \partial z_2} \Phi_2^{1/m}(z_1,z_2;\rho)$ with $\rho=-0.5$ and $m=2,10$. Such a function would always be positive for any value of $m>1$ if $\Phi_2(\cdot;-0.5)$ were max-id, but Figure~\ref{fig:gaussian} reveals large areas with negative values, especially for large $m$.
\begin{figure}[t!]
	\centering
	\includegraphics[width=0.35\linewidth]{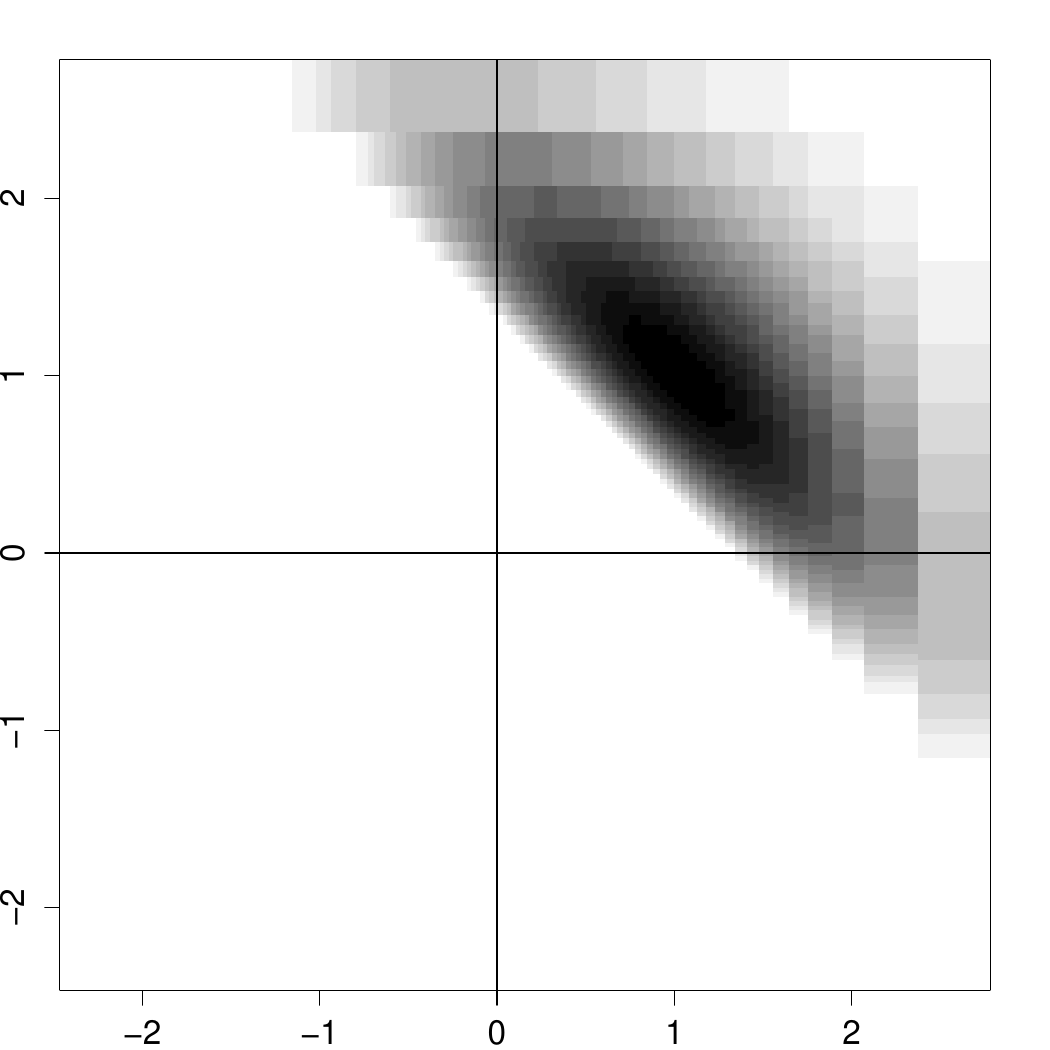} \quad
	\includegraphics[width=0.35\linewidth]{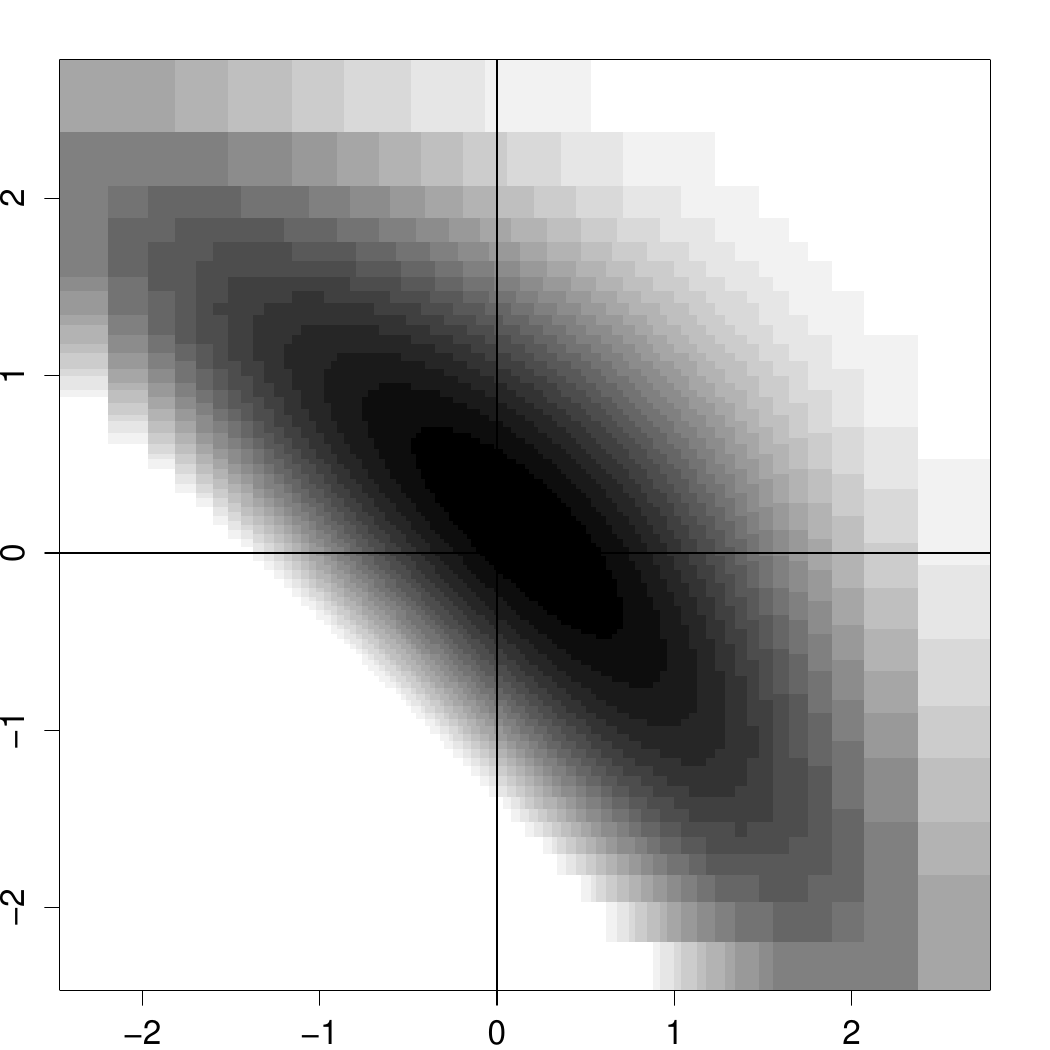}	
	%	\begin{minipage}{0.49\linewidth}
	%	\figurebox{17pc}{}{}[negdensgaussalpha05.eps]
	%	\end{minipage}\hfill\begin{minipage}{0.49\linewidth}
	%	\figurebox{17pc}{}{}[negdensgaussalpha01.eps]
	%	\end{minipage}
	%\figurebox{16.2pc}{}{}[negdensgaussalpha0501.eps]
	\caption{Function $\frac{\partial^2}{\partial z_1 \partial z_2} \Phi_2^{1/m}(z_1,z_2;\rho)$ for a bivariate standard Gaussian distribution $\Phi_2$ with correlation coefficient $\rho=-0.5$, for $m=2$ (left) and $m=10$ (right). The grey region corresponds to negative values, where darker areas indicate higher absolute values.}
	\label{fig:gaussian}
\end{figure}

\subsection{A dependence summary for asymptotic independence}\label{sec:eta}
We use the coefficient of tail dependence $\eta$ \citep{Ledford.Tawn.1996} to characterize faster bivariate joint tail decay as compared to marginal tails in the case of asymptotic independence.
If $\Pr(X_i>x)\sim 1/x$, $x\rightarrow\infty$ ($i=1,2$), then we assume that the following flexible joint tail representation holds along the diagonal:
\begin{equation}
\Pr(X_1>x,X_2>x)=\ell(x) x^{-1/\eta}, \qquad x\to\infty,
\end{equation}
with the coefficient of tail dependence $\eta\in(0,1]$ and a positive function $\ell$, slowly varying at infinity. Asymptotic independence arises if $\eta<1$, or if $\eta=1$ and $\ell(x)\rightarrow 0$, $x\rightarrow\infty$, while asymptotic dependence always implies $\eta=1$. Near-independence occurs when $\eta=1/2$. Positive and negative association in the tail correspond to $\eta\in(1/2,1]$ and $\eta\in(0,1/2)$, respectively.

\subsection{Properties of Gaussian-based models}\label{SMsec:theory}
The following Proposition~\ref{prop:welldef} shows that our Gaussian-based models, constructed from  $\kappa_{\bm\gamma}^{[k]}$ ($k=1,2,3$) and Gaussian spectral functions $W_i(\bm s)$ in Equation~(12) of \cite{ourpaper}, are well defined; in other words, the $D$-dimensional marginal intensity measures are locally  finite on $E=[-\infty,\infty]^D\setminus\{\bm l\}$ with $\bm l=(-\infty,\ldots,-\infty)^{\rm T}$, and the intensity measure of the set of points with components $+\infty$ is $0$.

\begin{proposition}[Well-definedness]\label{prop:welldef}\ 
	The construction $Z(\bm s)=\max_{i=1,2,\ldots} R_iW_i(\bm s)$, where $\{R_i\}$ are points of a Poisson process with mean measure $\kappa_{\bm\gamma}^{[k]}$ ($k=1,2,3$) in Equations~(16), (17) and (18) of \cite{ourpaper} respectively, and  $\{W_i(\bm s)\}$ are independent copies of a standard Gaussian process independent of $\{R_i\}$,  yields a well-defined max-id process.
\end{proposition}
\begin{proof}
	Observe that $\kappa_{\bm\gamma}^{[k]}$ $(k=1,2,3)$ are infinite measures on the positive half-line and that $\pr\{W_i(\bm s)>0\}>0$, such that $\pr\{Z(\bm s)>0)\}=1$ and we can focus on positive values of $Z(\bm s)$. The construction of $Z(\bm s)$ as the pointwise maximum over a Poisson process yields a valid max-id process, provided that all values of the multivariate exponent function, formally defined as $V(\bm z)=\Lambda_{\bm\gamma}^{[k]}([-\bm\infty,\bm z]^C)=\int_0^\infty \{1-F_{\bm W}(\bm z/r)\} \kappa_{\bm\gamma}^{[k]}(\mathrm{d}r)$, $\bm z\in (0,\infty)^D$ ($k=1,2,3$), are finite for any standard $D$-dimensional Gaussian vector $\bm W$. Using Mill's ratio for the univariate standard Gaussian density $\phi$, the multivariate standard Gaussian tail probability $1-F_{\bm W}(\bm x)$ can be bounded from above by $(1+\varepsilon_{\bm W})D\phi(\min_{j=1,\ldots,D} x_j) \min_{j=1,\ldots,D} x_j$ as $\min x_j\rightarrow \infty$ with a suitably fixed $\varepsilon_{\bm W}>0$. Therefore,  we can fix a constant $c_{\bm W}>0$ such that $1-F_{\bm W}(\bm x)<  c_{\bm W}\phi(\min_{j=1,\ldots,D} x_j) \min_{j=1,\ldots,D} x_j$ for all $\bm x$ with $\min_{j=1,\ldots,D} x_j > 0$. In the decomposition $V(\bm z)=\int_0^1 \{1-F_{\bm W}(\bm z/r)\} \kappa_{\bm\gamma}^{[k]}(\mathrm{d}r)+\int_1^\infty \{1-F_{\bm W}(\bm z/r)\} \kappa_{\bm\gamma}^{[k]}(\mathrm{d}r)$, the second term on the right hand side is finite, and it remains to prove finiteness of the first term. By using the upper bound on the multivariate Gaussian tail probability, collecting all constant terms in a constant $C$ and writing $z_m=\min_{j=1,\ldots,D} z_j$, we get 
	%$\int_0^1 \{1-F_{W}(z/r)\} \kappa_{\bm\gamma}^{[k]}(\mathrm{d}r) \leq Cm\int_0^1 \exp\{-m^2/(2r^2)\} r^{-1}f^{[k]}(r)\,\mathrm{d}r \stackrel{r\rightarrow 1/r}{=} Cm\int_{1}^\infty \exp(-m^2r^2/2)r^{-1}f^{[k]}(1/r)\,\mathrm{d}r$
	\begin{align*}
	\int_0^1 \{1-F_{\bm W}(\bm z/r)\} \kappa_{\bm\gamma}^{[k]}(\mathrm{d}r) &\leq Cz_m\int_0^1 \exp\{-z_m^2/(2r^2)\} r^{-1}f^{[k]}(r)\,\mathrm{d}r \\
	&\stackrel{r\rightarrow 1/r}{=} Cz_m\int_{1}^\infty \exp(-z_m^2r^2/2)r^{-1}f^{[k]}(1/r)\,\mathrm{d}r.
	\end{align*}
	The term $\exp(-z_m^2r^2/2)$ ascertains the tail decay of the integrand to be  faster than exponential  (as $r^{-1}f^{[k]}(1/r)$ has a polynomial tail), ensuring the finiteness of the upper bound and therefore of $V(\bm z)$.
\end{proof}

The following proposition proves the Weibull tail decay in the radial variables of  the elliptically contoured intensity of our Gaussian-based models; see Proposition~1 in \citet{ourpaper} for the elliptical representation. A consequence of this result is that the contribution of the points $R_i$ with $R_i\leq 1$ can be neglected in the asymptotic analysis of the tail behavior of max-id random vectors of this model family.
\begin{proposition}[Weibull tail decay]\label{prop:weibtail}\
	Under the assumptions of Proposition~1 of \cite{ourpaper}, the intensity measures $\tilde{\kappa}_{\bm\gamma}^{[k]}$ (see Equation~(23)  of \cite{ourpaper}) of the radial Poisson process $\{\tilde{R}_i,\ i=1,2,\ldots\}$ ($k=1,2,3$) have univariate Weibull tail with Weibull coefficient $2\beta/(\beta+2)$. 
\end{proposition}

\begin{proof}
	For ease of notation, we omit the superscript $[k]$ and subscript $\bm\gamma$ in $\tilde{\kappa}_{\bm\gamma}^{[k]}$ in the following and  denote the distribution function of $R_{\bm W}$ by $F_{\chi_D}$. We write the tail measure of $\tilde{\kappa}$ as  $\tilde{\kappa}([z,\infty)) =\tilde{\kappa}_1([z,\infty))+\tilde{\kappa}_2([z,\infty))$, where 
	$$\tilde{\kappa}_1([z,\infty))=\int_{0}^1\overline{F}_{\chi_D}(z/r) f^{[k]}(r)\,\mathrm{d}r,\qquad \tilde{\kappa}_2([z,\infty))=\int_1^\infty\overline{F}_{\chi_D}(z/r) f^{[k]}(r)\,\mathrm{d}r.$$
	The measure $\tilde{\kappa}_1$ has an infinite mass over $(0,\infty)$ due to the infinite number of Poisson points $R_i\leq 1$, while the measure $\tilde{\kappa}_2$ has a finite mass that corresponds to the points $R_i>1$. To prove the Weibull  tail behavior of $\tilde{\kappa}$, we first show that $\tilde{\kappa}_2$ is Weibull-tailed with  Weibull coefficient $2\beta/(\beta+2)$, and we then show that the tail of  $\tilde{\kappa}_1$ is asymptotically dominated by the one of $\tilde{\kappa}_2$ as $z$ tends to infinity. Notice  that $\tilde{\kappa}_1([z,\infty))/\tilde{\kappa}_2([z,\infty))\rightarrow 0$ for $z\rightarrow \infty$ if the Weibull coefficients $\beta_1$ and $\beta_2$  of $\tilde{\kappa}_1$ and $\tilde{\kappa}_2$ respectively satisfy $\beta_1>\beta_2$. The intensity measure $\tilde{\kappa}_2$ can be represented as $c\tilde{H}$ with $c=\tilde{\kappa}_2([1,\infty))>0$ and $\tilde{H}$ a probability distribution.  Based on results for the product of Weibull-type random variables \citep{ha2014,Huser.etal:2017}, one easily shows that $\tilde{H}$ is Weibull-tailed with coefficient $2\beta/(\beta+2)$, where $2$ is the Weibull coefficient of $F_{\chi_D}$. Therefore,  $\tilde{\kappa}_2$ is Weibull-tailed with coefficient $2\beta/(\beta+2)$. 
	%The Poisson process$\{R_iW_{S,i}, R_i>1 \}$  associated  with  $\tilde{\kappa}_2$  characterizes a max-id construction with a finite Poisson measure as proposed in Section \ref{sec:finitemeasure} with $c=\kappa^{[k]}[1,\infty]>0$ and $H$ a $d$-variate elliptical distribution, whose radial variable has distribution $H(z)=c^{-1}\tilde{\kappa}[z,\infty]\, 1_{[1,\infty)}(z)$. 
	%Based on results on the product of Weibull-type random variables \citep{ha2014,Huser.etal:2017}, one easily shows that the radial variable in the elliptical representation of the distribution $H$ has Weibull-type tail with coefficient $2\beta/(\beta+2)$, where $2$ is the Weibull coefficient of the $\chi$-distributed random variables $R_{S,i}$. Using the max-id tail approximation \eqref{eq:approxtail} relating the tails of the multivariate max-id distribution with $H$, we conclude that the distribution of the max-id random vector constructed from the componentwise maximum
	% \begin{equation*}
	% \tilde{Z}_{S,1} = \max_{R_i >1} R_iW_{S,i}
	% \end{equation*}
	% is Weibull-type with coefficient $2\beta/(\beta+2)$. 
	To show that the tail of $\tilde{\kappa}_1$ is lighter such that its contribution can be neglected,  we now fix an arbitrary small $0<\varepsilon_1<2$ and a constant $C_1>0$ such that   $\overline{F}_{\chi_D}(r)\leq C_1 \sqrt{2\pi}^{-1}\exp(-r^{2-\varepsilon_1})$ for all $r\leq 1$. 
	%Using the fact that $r^{-m}f^{[k]}(r)$ is positive and  upper-bounded by some constant $C_2$ for any fixed $m>0$, 
	%We develop $\tilde{\kappa}_1[z,\infty)$ as follows:
	Then, 
	\begin{align*}
	\tilde{\kappa}_1([z,\infty)) &\leq C_1\int_0^1  \exp\{-(z/r)^{2-\varepsilon_1}\} f^{[k]}(r)\,\mathrm{d}r \\
	&\stackrel{r\rightarrow 1/r}{=}   C_1\int_1^\infty  %(2-\varepsilon_1)^{-1}r^{2-\varepsilon_1} 
	\exp\left(-z^{2-\varepsilon_1}r^{2-\varepsilon_1}\right)f^{[k]}(1/r)r^{-2}\,\mathrm{d}r.
	\end{align*}
	For  any parameter values of $\alpha$ and $\beta$ in the construction of $\tilde{\kappa}_{\bm\gamma}^{[k]}$ and for $z>z_0>0$ with some fixed $z_0$, we can fix constants $\varepsilon_2>0$ with $\varepsilon_1+\varepsilon_2<2$ and $C_2>0$ such that  
	$$C_1\exp\left(-z^{2-\varepsilon_1}r^{2-\varepsilon_1}\right)f^{[k]}(1/r)r^{-2}\leq C_2(2-\varepsilon_1-\varepsilon_2)\exp\left(-z^{2-\varepsilon_1}r^{2-\varepsilon_1-\varepsilon_2}\right) z^{2-\varepsilon_1} r^{1-\varepsilon_1-\varepsilon_2}$$
	for $r\geq1$. 
	%	\begin{equation*}
	%	C_1\exp\left(-z^{2-\varepsilon_1}r^{2-\varepsilon_1}\right)f^{[k]}(1/r)r^{-2}\leq C_2(2-\varepsilon_1-\varepsilon_2)\exp\left(-z^{2-\varepsilon_1}r^{2-\varepsilon_1-\varepsilon_2}\right) z^{2-\varepsilon_1} r^{1-\varepsilon_1-\varepsilon_2}, \  r\geq 1.
	%	\end{equation*}
	This result yields the following upper bound: %$\tilde{\kappa}_1([z,\infty)) \leq C_2 \int_1^\infty  (2-\varepsilon_1-\varepsilon_2) \exp\left(-z^{2-\varepsilon_1}r^{2-\varepsilon_1-\varepsilon_2}\right) z^{2-\varepsilon_1} r^{1-\varepsilon_1-\varepsilon_2}\,\mathrm{d}r =-C_2\exp\left(-z^{2-\varepsilon_1}r^{2-\varepsilon_1-\varepsilon_2}\right)\big\vert_{r=1}^{r=\infty}\;=\;C_2\exp\left(-z^{2-\varepsilon_1}\right)$.
	\begin{align*}
	\tilde{\kappa}_1[z,\infty) &\leq C_2 \int_1^\infty  (2-\varepsilon_1-\varepsilon_2) \exp\left(-z^{2-\varepsilon_1}r^{2-\varepsilon_1-\varepsilon_2}\right) z^{2-\varepsilon_1} r^{1-\varepsilon_1-\varepsilon_2}\,\mathrm{d}r \\ &=-C_2\exp\left(-z^{2-\varepsilon_1}r^{2-\varepsilon_1-\varepsilon_2}\right)\big\vert_{r=1}^{r=\infty}\;=\;C_2\exp\left(-z^{2-\varepsilon_1}\right).
	\end{align*}
	%By choosing $\varepsilon_1$ small enough, we get that $2-\varepsilon_1>2\beta/(\beta+2)$, so that the Weibull-type tail of $\tilde{\kappa}_2$ dominates $\tilde{\kappa}_1$, and therefore $\tilde{\kappa}$ is Weibull-tailed with Weibull coefficient $2\beta/(\beta+2)$.  
	For small $\varepsilon_1$ such that $2-\varepsilon_1>2\beta/(\beta+2)$, the Weibull-type tail of $\tilde{\kappa}_2$ dominates $\tilde{\kappa}_1$. Thus, $\tilde{\kappa}$ is Weibull-tailed with Weibull coefficient $2\beta/(\beta+2)$.
\end{proof}

\subsection{Coverage probabilities for confidence intervals based on asymptotic normality}\label{SMsec:coverage}
Here, we extend our simulation study for the parsimonious max-id model with infinite exponent measure (recall equation (18) of \citet{ourpaper}), and we study the coverage probabilities of confidence intervals for the vector of parameters $\log\bm \psi=(\log\beta,\log \lambda)^{\rm T}$. The parameter $\beta\in[0,\infty)$ controls how far the model is from the Schlather max-stable model (with $\beta=0$ corresponding to max-stability), while $\lambda>0$ is a range parameter. In our Monte Carlo study, we simulate data at $D=10$ sites randomly generated on the unit square $[0,1]^2$, with $n=50$ independent replicates. We then estimate the parameters $\beta$ and $\lambda$ by maximizing the pairwise likelihood (20) in \citet{ourpaper} with cut-off distance $\delta=0.5$, and we compute confidence intervals for $\log\bm \psi=(\log\beta,\log \lambda)^{\rm T}$ based on asymptotic normality and the asymptotic variance formula $n^{-1}J^{-1}KJ^{-1}$ with matrices $J$ and $K$ estimated as in (28) of \citet{ourpaper}. We then replicate this experiment 1000 times, and we compute coverage probabilities for the nominal probabilities $0.05,0.10,0.15,\ldots,0.90,0.95$. The results are shown in Figure~\ref{fig:coverage}.

\begin{figure}[h!]
	\centering
	\includegraphics[width=.9\linewidth]{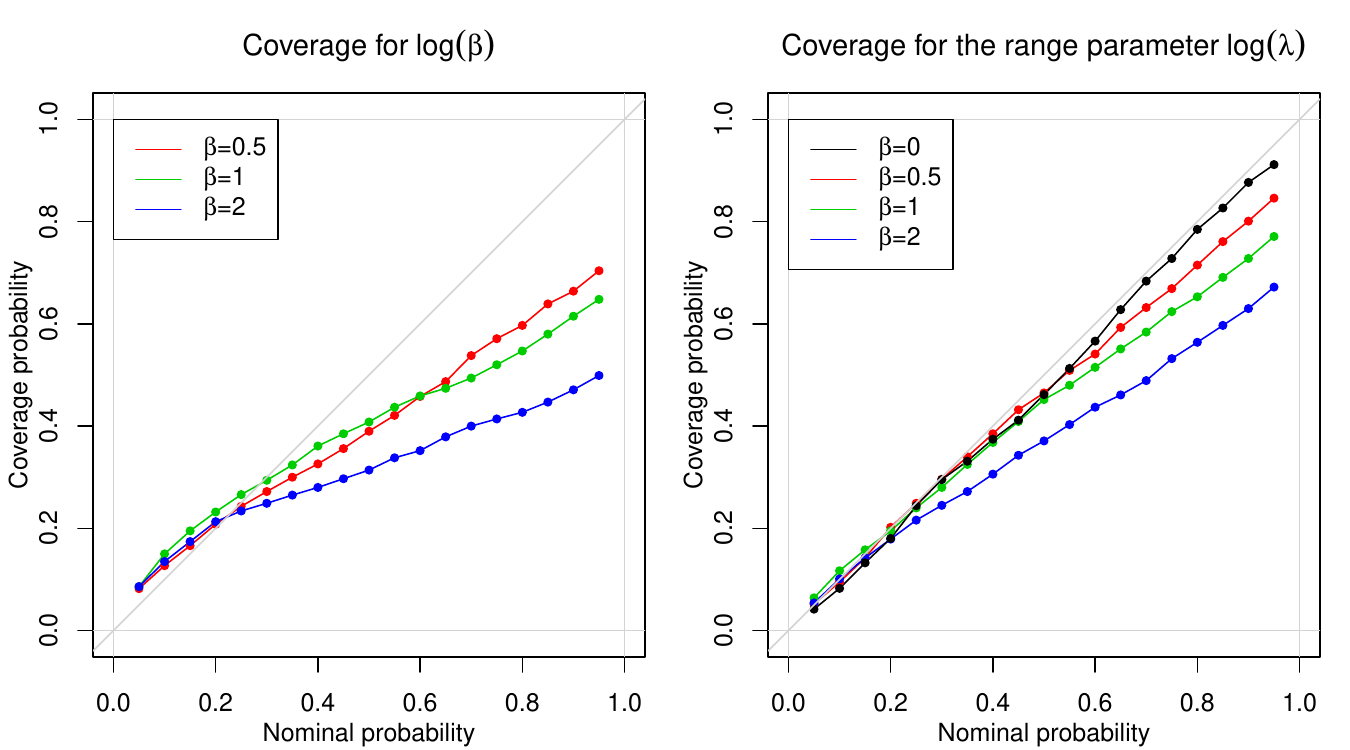}	
	%\figurebox{29pc}{}{}[Boxplot_sim2_D20.eps]
	\caption{Coverage probabilities for the confidence intervals for $\log\bm \psi=(\log\beta,\log \lambda)^{\rm T}$ based on asymptotic normality and the asymptotic variance formula $n^{-1}J^{-1}KJ^{-1}$ with matrices $J$ and $K$ estimated as in (28) of \citet{ourpaper}. When $\beta=0$, we only show the coverage for $\log\lambda$ for obvious reasons.}
	\label{fig:coverage}
\end{figure}

In all cases, the coverage probabilities are too small, indicating a tendency to underestimate the uncertainty. For the range parameter, $\log\lambda$, this underestimation is quite moderate, giving quite decent coverages overall. For $\log\beta$, however, the underestimation is stronger, especially for large $\beta$, indicating that uncertainty results need to be interpreted with care. We explain this uncertainty of coverage probabilities by the low sample size ($n=50$) suggesting that the estimator $\hat{\bm\psi}$ may not have reached its asymptotically normal distribution, and by the potential numerical instabilities when estimating the sandwich matrices $J$ and $K$. Another potential explanation for this phenomenon is the correlation between estimated parameters $\log\hat\beta$ and $\log\hat\lambda$. An alternative approach for estimating the uncertainty could be to use the bootstrap, but this would be extremely demanding for some of our models (especially the infinite measure max-id models), preventing its use in our context.

\subsection{Further details for the Dutch wind gust application}\label{SMsec:applic}
We here show quantile-quantile (QQ)-plots for the marginal fits in the data application of \cite{ourpaper}. 
In Figures~\ref{fig:qq1}--\ref{fig:qq8}, we provide the location-wise QQ-plots with data transformed to the unit Fr\'echet scale according to the fitted marginal model, showing a satisfactory fit. Finally, Figure~\ref{fig:qqall} shows boxplots for all locations and time steps pulled together, for daily, weekly, monthly and yearly maxima. Again, the fits look satisfactory. %Note that we do not provide location-wise QQ-plots for yearly maxima since the sample size is very small. 

\begin{figure}[h!]
	\centering
	\includegraphics[width=.9\linewidth]{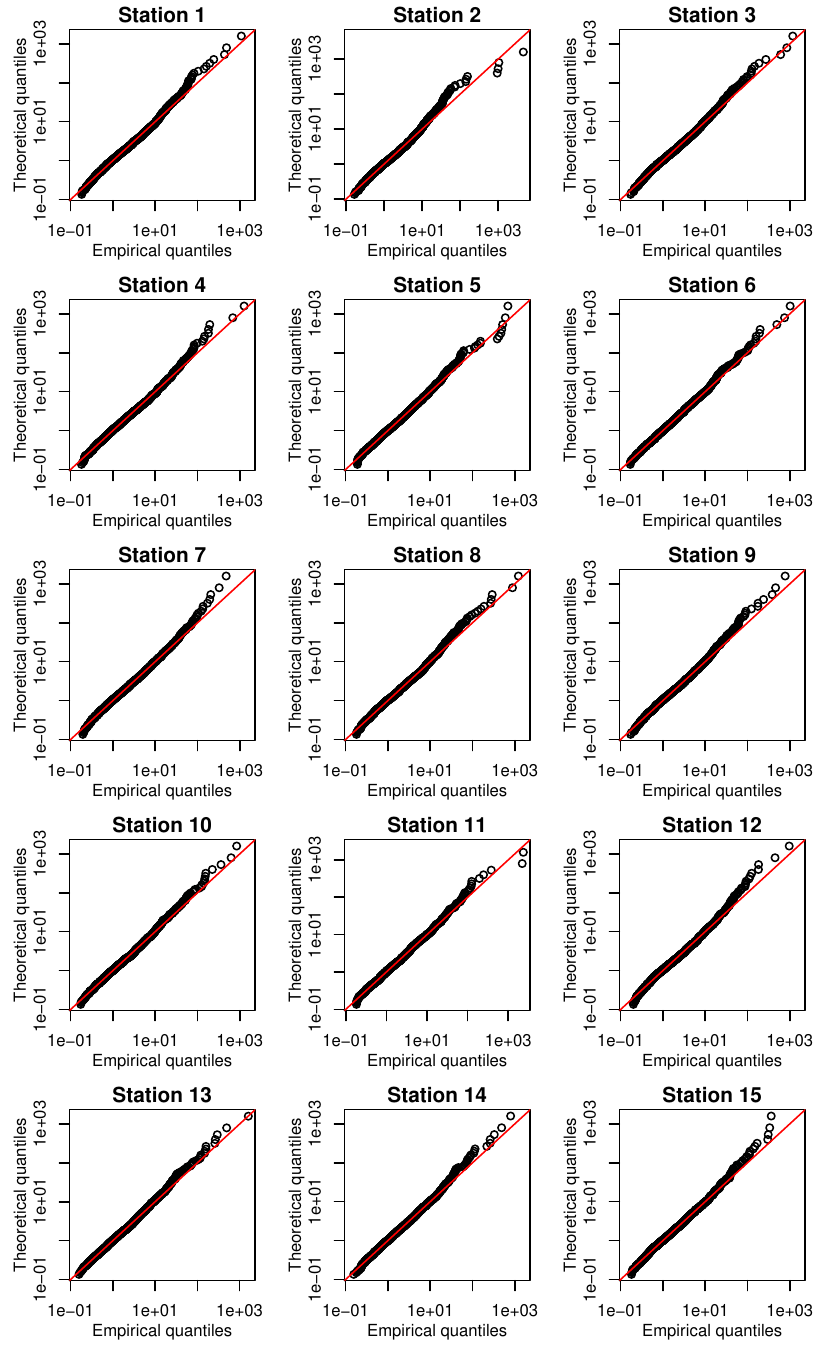}	
	%\figurebox{29pc}{}{}[Boxplot_sim2_D20.eps]
	\caption{QQ-plots of station-wise daily maxima. }
	\label{fig:qq1}
\end{figure}

\begin{figure}[h!]
	\centering
	\includegraphics[width=.9\linewidth]{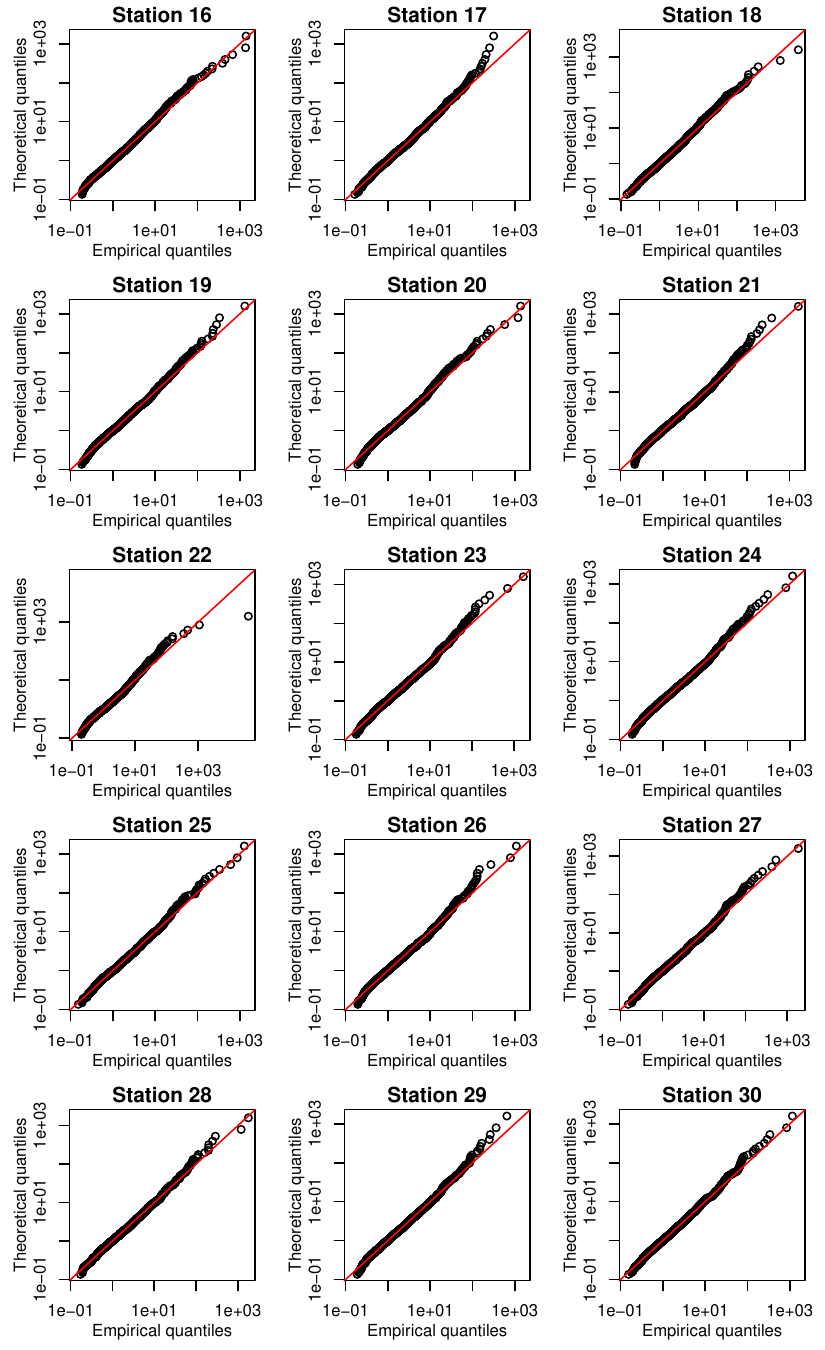}	
	%\figurebox{29pc}{}{}[Boxplot_sim2_D20.eps]
	\caption{QQ-plots of station-wise daily maxima. }
	\label{fig:qq2}
\end{figure}

\begin{figure}[h!]
	\centering
	\includegraphics[width=.9\linewidth]{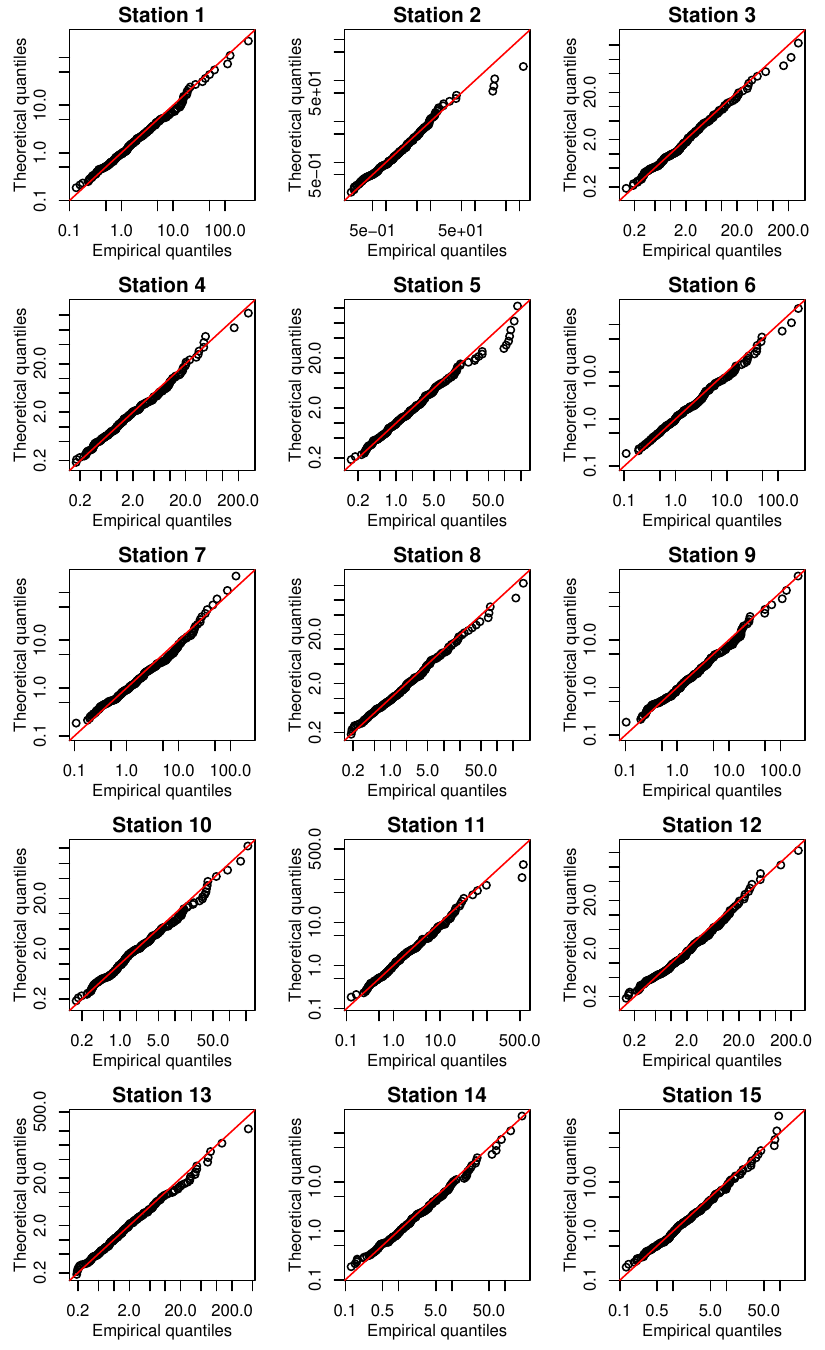}	
	%\figurebox{29pc}{}{}[Boxplot_sim2_D20.eps]
	\caption{QQ-plots of station-wise weekly maxima. }
	\label{fig:qq3}
\end{figure}

\begin{figure}[h!]
	\centering
	\includegraphics[width=.9\linewidth]{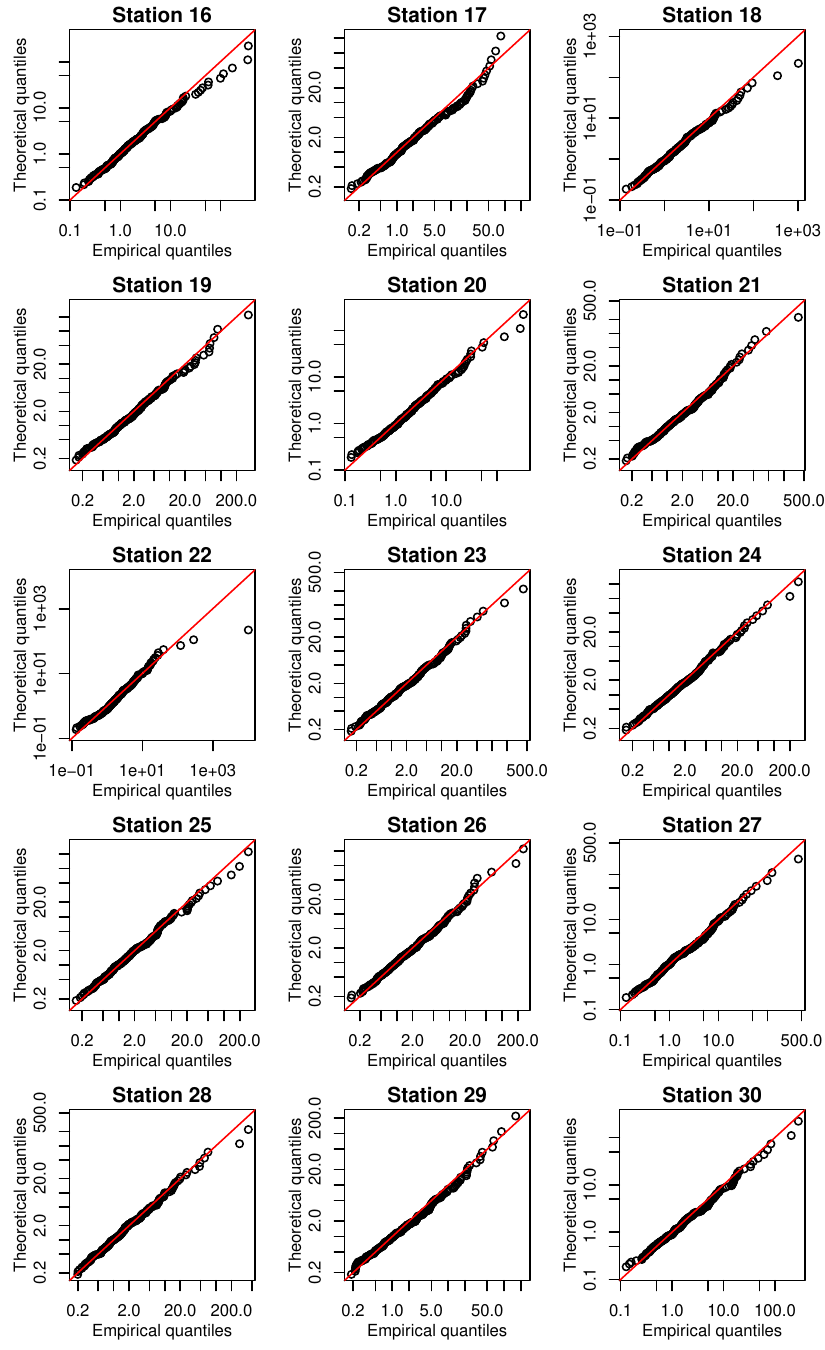}	
	%\figurebox{29pc}{}{}[Boxplot_sim2_D20.eps]
	\caption{QQ-plots of station-wise weekly maxima. }
	\label{fig:qq4}
\end{figure}

\begin{figure}[h!]
	\centering
	\includegraphics[width=.9\linewidth]{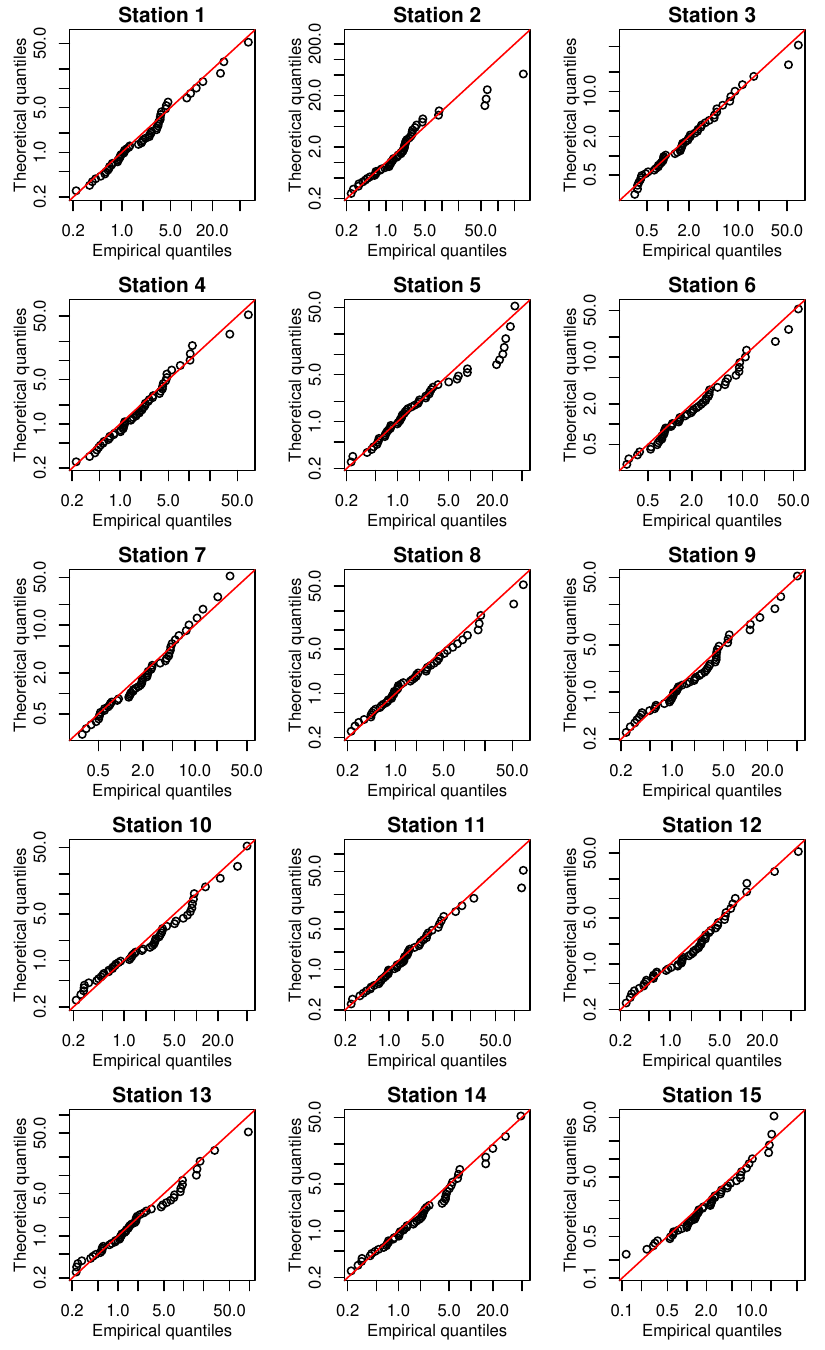}	
	%\figurebox{29pc}{}{}[Boxplot_sim2_D20.eps]
	\caption{QQ-plots of station-wise monthly maxima. }
	\label{fig:qq5}
\end{figure}

\begin{figure}[h!]
	\centering
	\includegraphics[width=.9\linewidth]{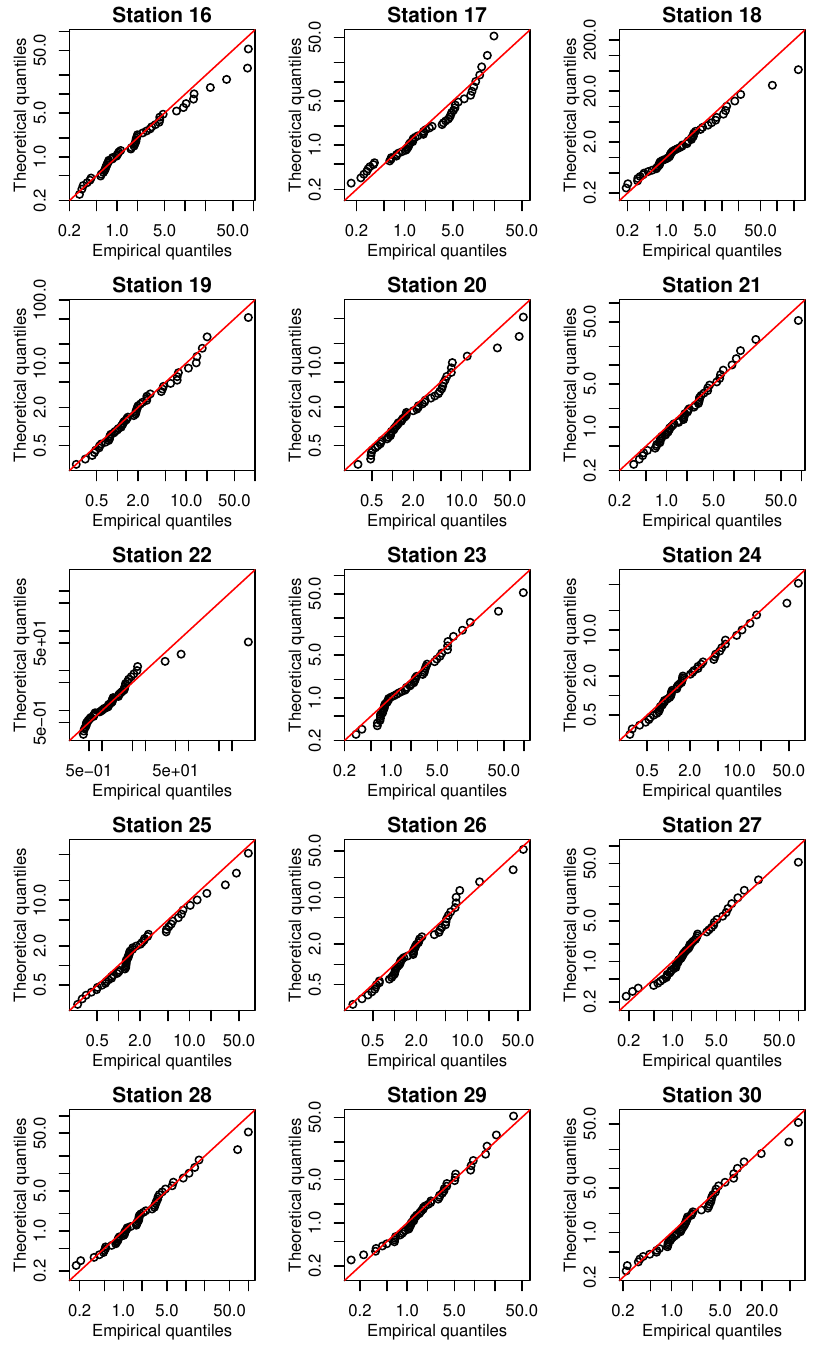}	
	%\figurebox{29pc}{}{}[Boxplot_sim2_D20.eps]
	\caption{QQ-plots of station-wise monthly maxima. }
	\label{fig:qq6}
\end{figure}

\begin{figure}[h!]
	\centering
	\includegraphics[width=.9\linewidth]{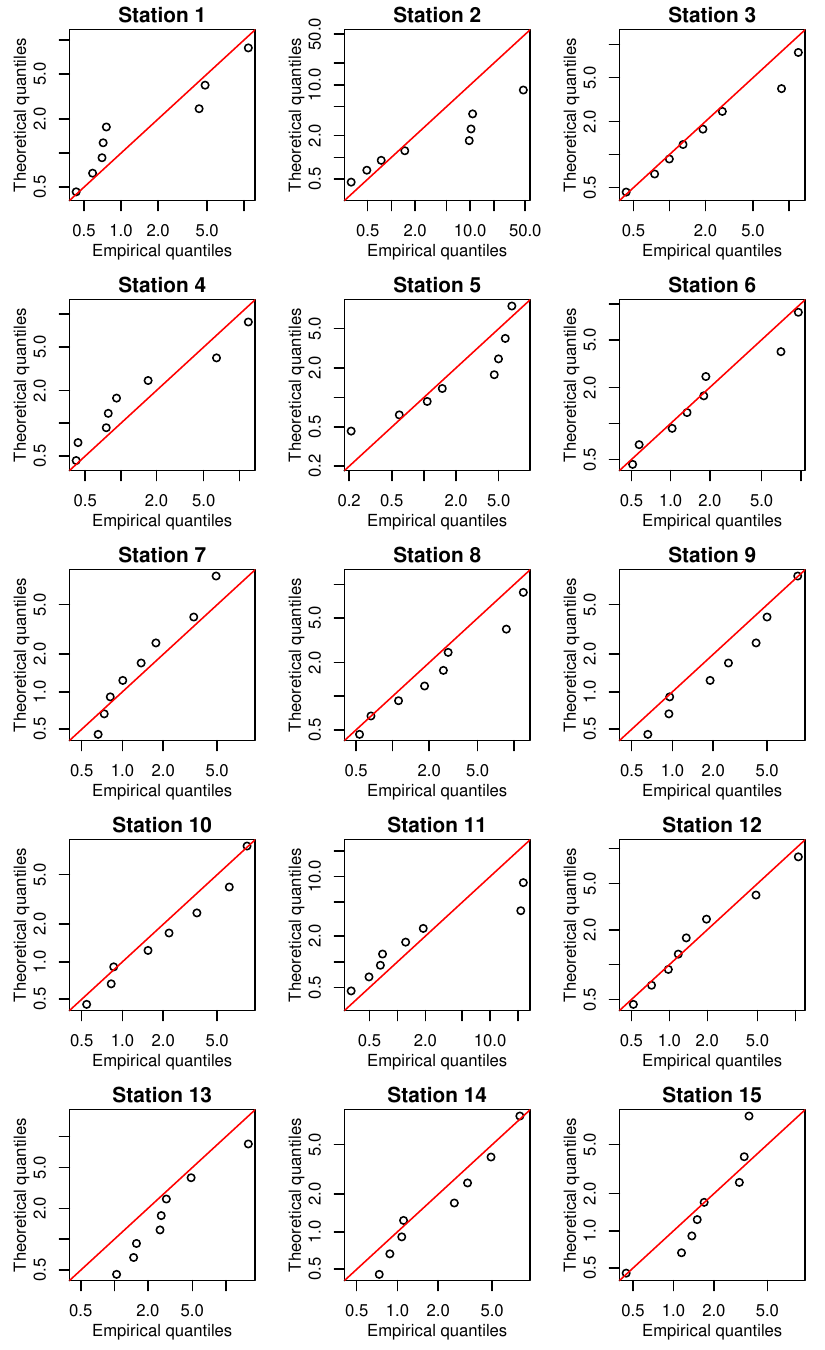}	
	%\figurebox{29pc}{}{}[Boxplot_sim2_D20.eps]
	\caption{QQ-plots of station-wise yearly maxima. }
	\label{fig:qq7}
\end{figure}

\begin{figure}[h!]
	\centering
	\includegraphics[width=.9\linewidth]{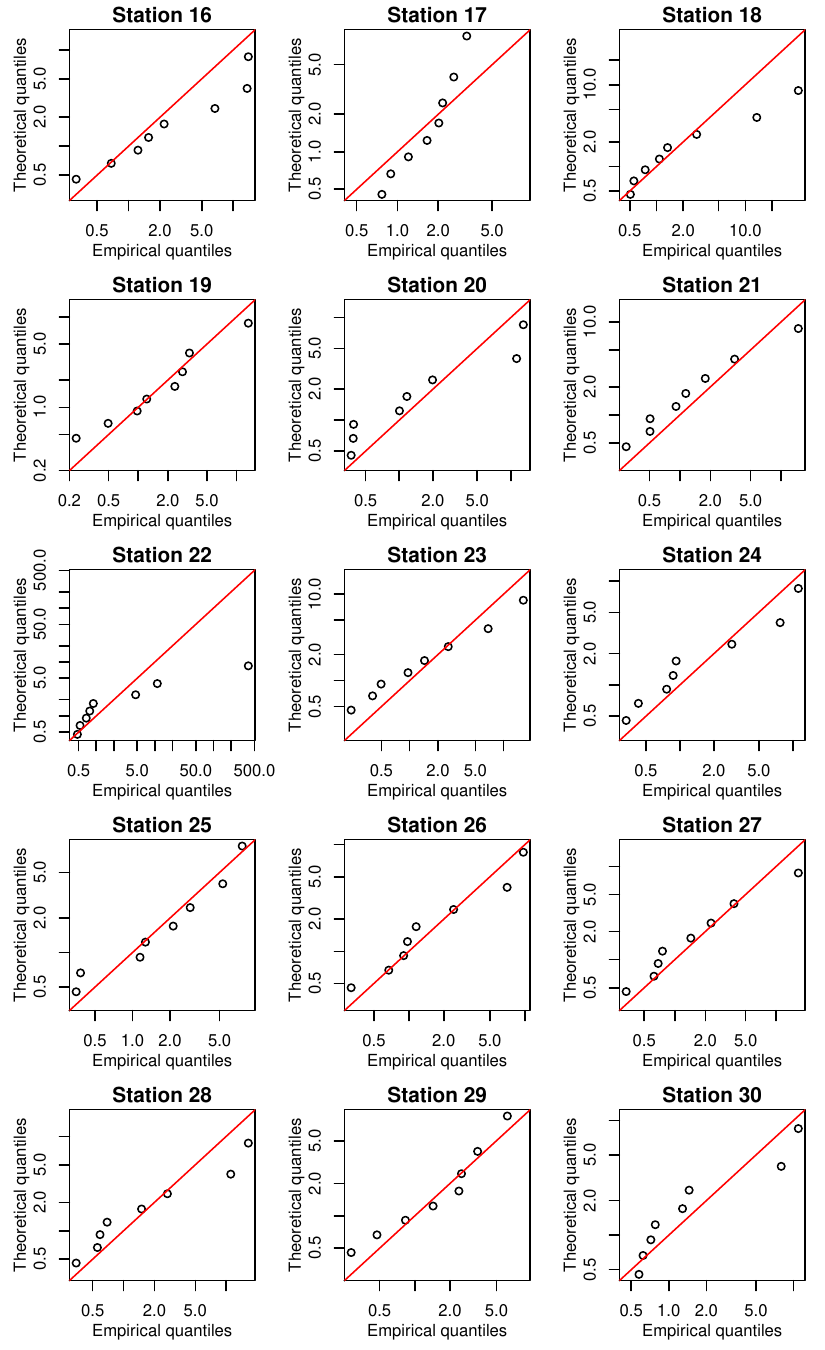}	
	%\figurebox{29pc}{}{}[Boxplot_sim2_D20.eps]
	\caption{QQ-plots of station-wise yearly maxima. }
	\label{fig:qq8}
\end{figure}

\begin{figure}[h!]
	\centering
	\includegraphics[width=.9\linewidth]{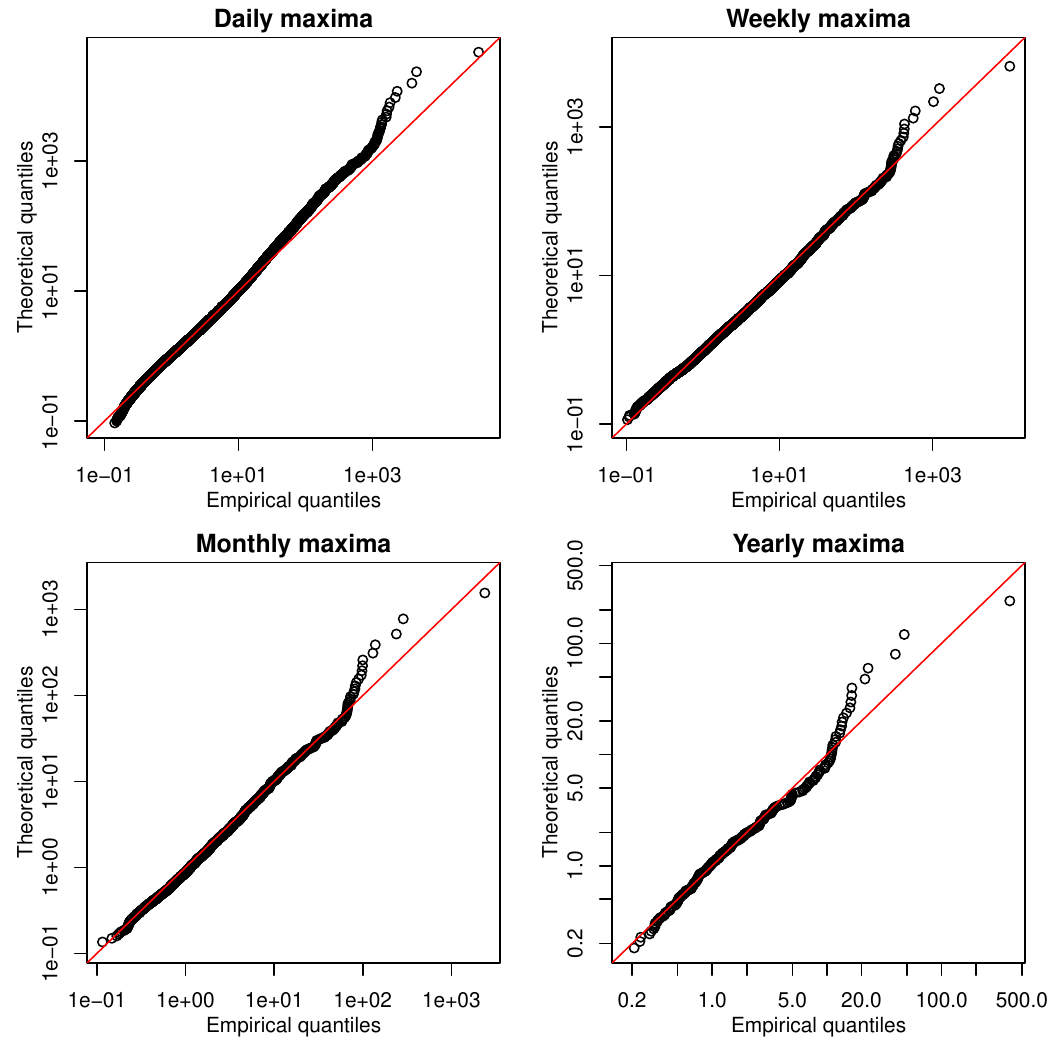}	
	%\figurebox{29pc}{}{}[Boxplot_sim2_D20.eps]
	\caption{QQ-plots of maxima pooled over all stations for daily (top left), weekly (top right), monthly (bottom left) and yearly (bottom right) maxima.}
	\label{fig:qqall}
\end{figure}

%%%%%%%%%%%%%%%%%%%%%%%%%%%%%%%%%
%%%%%%%%%%%%%%%%%%%%%%%%%%%%%%%%%
%%%%%%%%%%%%%%%%%%%%%%%%%%%%%%%%%

\baselineskip 20pt

\FloatBarrier

\bibliographystyle{CUP}
\bibliography{ref}

\baselineskip 10pt

\end{document}